\documentclass[a4paper,UKenglish,cleveref, autoref, thm-restate]{lipics-v2021}

\usepackage[dvipsnames]{xcolor}
\usepackage{cancel}
\usepackage[utf8]{inputenc} 
\usepackage[T1]{fontenc}    
\usepackage{hyperref}       
\usepackage{url}            
\usepackage{booktabs}       
\usepackage{amsfonts}       
\usepackage{nicefrac}       
\usepackage{microtype}      
\usepackage{algorithmicx} 
\usepackage{lipsum}
\usepackage{wrapfig}

\usepackage{algorithm2e}

\usepackage{thm-restate}
\usepackage{marginnote}

\usepackage{ifoddpage}
\definecolor{vOne}{HTML}{7570B3}
\definecolor{vTwo}{HTML}{D95F02}
\definecolor{vThree}{HTML}{1B9E77}

\newcommand{\reviewer}[1]{%
  \checkoddpage
  \ifoddpage
    \normalmarginpar
  \else
    \reversemarginpar
  \fi
  \marginpar{\raggedright\itshape\footnotesize #1}%
}

\usepackage{amsmath,amssymb,wasysym }
\usepackage{amsthm}
\usepackage{graphicx}
\usepackage{tikz,comment}
\usepackage{pgfplots}

\newcommand{\Poiss}[1]{{{\tt  Poiss}({#1})}}
\newcommand{\Poissone}[1]{{{\tt  Poiss}_{\geq 1}({#1})}}

\newcommand{\stirlingsecond}[2]{\genfrac{\{}{\}}{0pt}{}{#1}{#2}}

\newcommand{\E}{{\mathbb{E}}}

\usepackage{pifont}

\title{Efficient sampling of the unlabeled structure  of uniform random mappings and surjections} 

\titlerunning{Efficient sampling of the unlabeled structure  of uniform random mappings and surjections} 

 \author{Arnaud Carayol}{Laboratoire d'Informatique Gaspard-Monge, 
  Université Gustave Eiffel, France}{arnaud.carayol@univ-eiffel.fr
 }{}{}

 \author{Pablo Rotondo}{Laboratoire d'Informatique Gaspard-Monge, 
  Université Gustave Eiffel, France}{pablo.rotondo@univ-eiffel.fr
 }{}{}

\authorrunning{A. Carayol and P. Rotondo} 

\Copyright{Arnaud Carayol and Pablo Rotondo} 



\nolinenumbers 

\begin{document}

\maketitle

\newcommand{\length}[2]{L_{#1,#2}}
\begin{abstract}
This article deals with the sampling and statistics of uniform random mappings and surjections from $[n]=\{1,\ldots,n\}$ to $[k]=\{1,\ldots,k\}$. More precisely, we focus on statistics that are invariant under relabeling of the domain and codomain. The structure of a mapping or surjection up to relabeling is captured by a compact representation called its \emph{profile}. The profile records, for each occurring preimage size, the number of preimages of that size.

We provide a complete set of algorithms that directly sample the profiles of uniform random mappings and surjections, covering the full parameter space. Their expected running times match the expected length of their outputs up to logarithmic factors. It is important to note that we make no assumption about the relationship between $n$ and $k$; in particular, we do not assume that $n/k$ is constant. All our guarantees are uniform in $n$ and $k$. Achieving this uniformity substantially complicates the analysis.  Our algorithms are based on exact-size Boltzmann samplers, which we accelerate by optimizing the cost of each sampling round and increasing the acceptance probability.

As spin-off results, we answer, up to a logarithmic factor, an open question of Devroye and Los~\cite{devroyelos25} on sampling profiles of random mappings when $n>k\log k$, and obtain a sampler for uniform random surjections with expected running time $O(n\log k)$, valid across the full parameter space.
\end{abstract}

\clearpage
\section{Introduction}

 Mappings and surjections from $[n]=\{1,\ldots,n\}$  to  $[k]=\{1,\ldots,k\}$ are fundamental objects in combinatorics and arise both explicitly and as building blocks in many constructions. For instance, random mappings describe the outcome of balls-into-bins processes. Surjections encode words of length $n$ over an alphabet of size $k$ in which every letter appears at least once. Up to a permutation of $[k]$, they correspond to set partitions of $[n]$ into $k$ nonempty parts.  More surprisingly, deterministic and complete finite-state automata in which all states are reachable from the initial state (accessible DFAs) can be efficiently sampled using random surjections (see \cite{bassino-nicaud07,nicaud-random-automata-mfcs2015}). 
 This article deals with the efficient sampling of  uniform  random mappings and surjections and their statistics invariant under relabeling.
 
\paragraph*{The profile of a mapping}

\begin{wrapfigure}{r}{0.4\linewidth}
  \centering
  \resizebox{0.95\linewidth}{!}{
    \begin{tikzpicture}[
        pt/.style={circle,fill=#1,inner sep=0pt,minimum size=5pt},
        pt/.default=black,
        edge/.style={#1,line width=0.8pt},
        num/.style={font=\scriptsize,label distance=1pt},
        x=0.5cm, y=1.8cm,
      ]

      \foreach \i/\j/\c in {1/3/ForestGreen,2/1/ForestGreen,3/4/BrickRed,4/3/ForestGreen,5/1/ForestGreen,
                            6/4/BrickRed,7/2/RoyalBlue,8/3/ForestGreen,9/1/ForestGreen}
        \draw[edge=\c] (\i-1,0) -- (\j-1,1);
      \foreach \i in {1,...,9}
        \node[pt,label={[num]below:$\i$}] at (\i-1,0) {};
      \foreach \j/\c in {1/ForestGreen,2/RoyalBlue,3/ForestGreen,4/BrickRed}
        \node[pt=\c,label={[num]above:$\j$}] at (\j-1,1) {};
      \node[left=4pt] at (-0.2,1) {$[k]$};
      \node[left=4pt] at (-0.2,0) {$[n]$};

      \begin{scope}[yshift=-3cm]
        
        \foreach \v/\a/\r/\c in {1/0/0/RoyalBlue, 2/1/1.5/BrickRed, 3/3/4/ForestGreen, 3/6/7/ForestGreen}{
          \foreach \d in {1,...,\v}{
            \draw[edge=\c] (\a+\d-1,0) -- (\r,1);
            \node[pt] at (\a+\d-1,0) {};
          }
          \node[pt=\c] at (\r,1) {};
        }
        \node[left=4pt] at (-0.2,1) {$[k]$};
        \node[left=4pt] at (-0.2,0) {$[n]$};
      \end{scope}
    \end{tikzpicture}
  }
  \caption{A surjection from $[9]$ onto $[4]$ with profile
    $(1,1),(2,1),(3,2)$ (top) and its unlabeled representation (bottom).}
  \label{fig:profile}
\end{wrapfigure}
We consider mappings up to relabeling of both the domain and codomain. Up to these relabelings, a mapping is characterized by what
we call its \emph{profile}.
Intuitively, the profile records, for each preimage size that occurs, how many elements of $[k]$ have a preimage of that size.
Formally, let $f\colon[n]\to[k]$ be a mapping. We define the \emph{size-vector} of $f$
as the vector $(s_1,\ldots,s_k)$ where $s_i=|f^{-1}(i)|$ is the size
of the preimage of $i$.
The \emph{profile} of $f$ is the sequence of pairs
\[
  (v_1,m_1),\ (v_2,m_2),\ \ldots,\ (v_\ell,m_\ell)
\]
where $v_1<\ldots<v_\ell$ are the distinct values occurring in the
size-vector and $m_j\geq 1$ is the number of indices $i\in[k]$ such
that $s_i=v_j$. The integer $\ell$ is the \emph{length} of the profile.
Note that $\sum_{j=1}^\ell m_j = k$ and
$\sum_{j=1}^\ell v_j\,m_j = n$.

Two mappings have the same profile if and only if one can be obtained
from the other by relabeling the domain and codomain.
The profile therefore represents a mapping up to relabeling\footnote{The group $S_n\times S_k$, where $S_m$ denotes
the symmetric group on $[m]$, acts on the set of mappings by permuting
the domain and codomain. Each orbit corresponds uniquely to a
distinct profile.}. While a mapping is naturally represented by $n$ numbers
from $[k]$, requiring approximately $n\log_2 k$ bits (which matches
the entropy of a uniform random mapping), its profile can always
be represented using $O(\sqrt{n}\,(\log n+\log k))$ bits. Indeed,
the following lemma shows that the length of the profile is always
$O(\sqrt{n})$.

\begin{restatable}{lemma}{lengthprofilemapping}
  \label{prop:length-profile-mapping}
  For any mapping $f\colon [n]\to [k]$, the length of its profile
  is at most $1+\sqrt{2n}$.
\end{restatable}
\begin{proof}
    Let $(v_1,m_1),\ (v_2,m_2),\ \ldots,\ (v_\ell,m_\ell)$ be the profile of $f$. We want to show that $\ell\leq 1 + \sqrt{2n}$. Observe that $0\leq v_1<v_2<\ldots$ implies that $v_1\geq 0$, $v_2\geq 1$ and so on. Then $n=\sum_{j=1}^\ell v_j m_j \geq \sum_{j=1}^\ell v_j  \geq \sum_{j=0}^{\ell-1}j = \frac{\ell (\ell-1)}{2} \geq (\ell-1)^2/2$ holds and the result follows.
\end{proof}

\paragraph*{Objectives of the article}

Our goal is to sample the profiles\footnote{We consider the distribution on profiles induced by the
uniform distribution on mappings (resp.\ surjections). In other words,
the probability of a profile is proportional to the number of mappings
(resp.\ surjections) realizing it.} of uniform random mappings and
surjections \emph{directly}, that is, without generating the underlying
mapping or surjection. The naive approach, which draws a mapping
uniformly at random and then computes its profile, has a cost at least
linear in $n$, far beyond the size of the output. We aim instead for
samplers whose expected running time matches the expected length of
the output. Since any sampler must at least
write the profile it returns, the expected length of a random profile
is a lower bound on the expected running time of any sampler. Our
algorithms match this lower bound up to logarithmic factors, and are
optimal in this sense.

A crucial point in this work is that we do not make any assumption
on the asymptotic dependence between the size $n$ of the domain and
the size $k$ of the codomain. When dealing with two parameters, it is
common to assume an asymptotic relation between them, such as
$n \sim \lambda k$ or $n \sim k^{\lambda}$ for some fixed $\lambda$.
We do not make any such assumption, and our results hold on the full
parameter space $(n,k)$. This strengthens the applicability
of the results, but also greatly complicates the proofs.

\paragraph*{Main motivations}

There are two main motivations for efficiently sampling profiles.

Our first motivation is that the profile determines every statistic of a mapping or surjection
that is invariant under relabeling of both the domain and codomain.
Sampling profiles directly can therefore speed up simulations involving
such statistics. For instance, Devroye and Los~\cite{devroyelos25}
recently used profile sampling for uniform random mappings to accelerate
the simulation of load-balancing processes modeled by balls-into-bins
allocations.

Our second motivation is to use efficient profile sampling as a
stepping stone towards optimal sampling of uniform random
surjections from $[n]$ to $[k]$. The entropy of a uniform random
surjection can be shown to be $\Theta(n\log k)$, uniformly in $n$
and $k$. To our knowledge, obtaining a sampler with expected
complexity $O(n\log k)$ over the full parameter space had remained
open\footnote{As discussed in the \hyperref[par:related-work]{related work paragraph}, Bassino and
Sportiello~\cite{BassinoSportiello2018LineartimeES} achieve this bound
in the regime $n\sim\lambda k$ for fixed $\lambda>1$, but the hidden
constant grows exponentially with $\lambda$.}.

It is easy to see that sampling a uniform random surjection
conditioned on its profile is straightforward. First, a uniform
random permutation of $[k]$ assigns the preimage sizes prescribed
by the profile to the codomain elements. An independent uniform
random permutation of $[n]$ is then split into consecutive blocks
of these sizes, assigning the elements of each block to the
corresponding codomain element. With slightly more work, this
conditional sampling can be performed in expected time
$O(n\log k)$ 
uniformly over the full parameter space. Combined with the efficient profile
samplers developed in this article, this reconstruction yields samplers
for uniform random surjections that are optimal up to a constant factor,
uniformly over the full parameter space. 

\paragraph*{Framework for the complexity analysis}

As is common in the analysis of sampling algorithms
(see, for instance,~\cite{devroyelos25}), we analyze the complexity
of our algorithms in an \emph{extended real arithmetic model}
(extended RAM), in which real numbers are stored exactly and
evaluations of standard mathematical functions have cost $O(1)$.
We additionally assume that special functions, such as the
Lambert $W$ function, the log-gamma function, and the regularized
incomplete gamma function,
can be evaluated in $O(1)$ time. This model also assumes that
sampling from the uniform distribution on $[0,1]$ has cost $O(1)$.
In this model, random variables from standard distributions,
such as Bernoulli, Poisson, and binomial distributions, can be
sampled in $O(1)$ expected time~\cite{devroye1986non}.

Even though we analyze our algorithms in an idealized setting,
they have been designed with practical implementability in mind.
In particular, arbitrary-precision requirements are confined
to a few black-box primitives, for which efficient approximation
schemes or sampling algorithms are
available~\cite{devroye1986non,Temme_1994}; see the discussion
in Section~\ref{sec:conclusion}. We believe that our algorithms
can be implemented and analyzed as exact samplers in the more
realistic word-RAM model, in which operations on words of
$\Theta(\log(n+k))$ bits and the generation of a fair random bit
have unit cost, with only a polylogarithmic overhead in $n$ and $k$.
Establishing this claim is beyond the scope of this article
and is left for future work.

This motivates our focus on complexity bounds up to multiplicative
polylogarithmic factors. We write $g(n,k)=\widetilde{O}(f(n,k))$
if there exists a polynomial $p$, independent of $n$ and $k$,
such that $g(n,k)\leq p(\log(n+k))\,f(n,k)$.
The notation $\widetilde{\Theta}$ is defined similarly.

\paragraph*{Statement of the main results}

In this section we present the main results for the profiles of random mappings and random surjections. For each case, we present a result for the expected length of the profile, followed by the complexity of the corresponding algorithm that is optimal up to polylogarithmic factors. We remark that in both cases there are essentially two regimes.

\newcommand{\Binomial}{{\tt Bin}}

\noindent
\textbf{Random mappings.} For a random mapping from $[n]$ to $[k]$, the sizes of the $k$ preimages belong to a window centred around $n/k$, of expected width\footnote{For random mappings, we allow for $k>n$, in particular $n/k$ can be arbitrarily small. This explains the extra term $1$ in $\mathbb{E}[\length{n}{k}]= \tilde{\Theta}(1+\sqrt{\tfrac{n}{k}} )$.} $\widetilde{\Theta}(1+\sqrt{n/k})$. Intuitively, near the center of this window, individual sizes occur with comparable proportions. This suggests two regimes as illustrated in Figure~\ref{fig:histograms}. First, when $k$ is smaller than the window size $\sqrt{n/k}$, equivalently $k\leq n^{1/3}$, repetitions among the $k$ sizes are limited and the expected profile length is in  $\Theta(k)$. Otherwise, when $k$ is greater than the window size $\sqrt{n/k}$, equivalently $k\geq n^{1/3}$, a significant part of the central window is occupied and the expected profile length is in $\widetilde{\Theta}(1+\sqrt{n/k})$. We make this intuition precise in Section~\ref{sec:profiles-random-mappings} by relating preimage sizes to Poisson random variables.

\begin{restatable}{proposition}{proplengthprofilesrandommapping}
    \label{prop:length-profiles-random}
The length $\length{n}{k}$ of the profile of  a random mapping from $[n]$ to $[k]$ satisfies:
\begin{itemize}
\item For $k \leq n^{1/3}$, $\mathbb{E}[\length{n}{k}] = \Theta(k)$; more precisely, $c_1 k \leq \mathbb{E}[\length{n}{k}] \leq k$;
\item For $k\geq n^{1/3} $, $\mathbb{E}[\length{n}{k}]= \tilde{\Theta}(1+\sqrt{\tfrac{n}{k}} )$;  $c_2(1+ \sqrt{\nicefrac{n}{k}}) \leq \mathbb{E}[\length{n}{k}] \leq c_3 (\sqrt{\nicefrac{n}{k} \log k } + \log k)$;

\end{itemize}
where $c_1, c_2, c_3$ are positive constants independent of $n$ and $k$.
\end{restatable}

    More precise results are possible in some regimes, but the above results  are  always valid and optimal up to logarithmic factors.

    \begin{figure}[h]
    \centering
\noindent

    \begin{tikzpicture}[scale=1, anchor=mid]
            \draw[->,thick] (0,0) --  (10,0);
                \draw (0,0.1) -- (0,-0.1);

        \node at (0,-0.5) {$1$};

        \fill[fill=blue!20] (0.05,0) rectangle (4.9,-0.1);

        \node at (10.1,.0) {$k$};
                
        \node at (2.5,.3) {\small $\Theta(k)$};
        \node at (2.5,-.4) {\footnotesize mostly distinct };

      \node at (4.9,-.5) {\tiny \color{blue}$\sqrt[3]{n}$};

      \draw[color=blue] (5,0.1) -- (5,-0.1);

        \fill[fill=blue!40] (5.1,0) rectangle (9.9,-0.1);

        \node at (7.5,.3) {\small $\tilde{\Theta}(1+\sqrt{\tfrac{n}{k}} )$}; 
        \node[align=center] at (7.5,-.5) {\footnotesize small window};

    \end{tikzpicture}

\begin{tikzpicture}[scale=0.4]
\begin{axis}[
    ybar interval,
    xlabel={Value},
    ylabel={Frequency},
        xtick={
        251170000,
        251180000,
        251190000,
        251200000,
        251210000
    },
    tick label style={
        /pgf/number format/fixed,
        /pgf/number format/precision=4
    },
    clip=false,
]
\addplot[fill=blue,
    fill opacity=0.2,
    draw opacity=0.6] coordinates {
(251169145.0, 88)
(251171891.25, 133)
(251174637.5, 171)
(251177383.75, 165)
(251180130.0, 226)
(251182876.25, 233)
(251185622.5, 250)
(251188368.75, 256)
(251191115.0, 299)
(251193861.25, 245)
(251196607.5, 276)
(251199353.75, 225)
(251202100.0, 202)
(251204846.25, 210)
(251207592.5, 155)
(251210338.75, 146)
(251213085.0, 124)
(251215831.25, 83)
};

\draw[thick,<->,color=red]
    (axis cs:251169145,270) -- (axis cs:251215831,270)
    node[midway, above] {\large $\approx 4.7\times 10^4$};
\end{axis}

\end{tikzpicture} 
\qquad\quad
\begin{tikzpicture}[scale=0.4]
\begin{axis}[
    ybar interval,
    xlabel={Value},
    ylabel={Frequency},
    xtick={
        995000,
        998000,
        1000000,
        1002000,
        1004000
    },
    tick label style={
        /pgf/number format/fixed,
        /pgf/number format/precision=3
    },
]
\addplot[fill=blue,
    fill opacity=0.4,
    draw opacity=0.6] coordinates {
(995053.0, 3)
(995306.075, 4)
(995559.15, 7)
(995812.225, 23)
(996065.3, 76)
(996318.375, 185)
(996571.45, 419)
(996824.525, 991)
(997077.6, 2017)
(997330.675, 4096)
(997583.75, 7335)
(997836.825, 12857)
(998089.9, 20684)
(998342.975, 31796)
(998596.05, 44570)
(998849.125, 59624)
(999102.2, 74500)
(999355.275, 88276)
(999608.35, 97481)
(999861.425, 100865)
(1000114.5, 97431)
(1000367.575, 89155)
(1000620.65, 75535)
(1000873.725, 61219)
(1001126.8, 46669)
(1001379.875, 32497)
(1001632.95, 21954)
(1001886.025, 13425)
(1002139.1, 7813)
(1002392.175, 4343)
(1002645.25, 2219)
(1002898.325, 1119)
(1003151.4, 471)
(1003404.475, 217)
(1003657.55, 73)
(1003910.625, 25)
(1004163.7, 16)
(1004416.775, 7)
(1004669.85, 1)
(1004922.925, 1)
};
\end{axis}
\end{tikzpicture}

\hfill
\caption{Sample histograms corresponding to a single random profile for both regimes. From left to right: $n=10^{12},k=\lfloor n^{0.3}\rfloor=3981$; $n=10^{12},k=n^{1/2}=10^6$.
\label{fig:histograms}}
\end{figure}

When $k$ is small, that is $k\leq n^{1/3}$, we propose a simple algorithm based on the multinomial method of Devroye, see Section~\ref{sec:profile-mapping-k-small} for the details.

\begin{restatable}{theorem}{thmrmmainrandommappingsi}
    \label{thm:random-mapping-1}
    Algorithm~\ref{algo:k-small} produces the profile of a uniform random mapping with expected cost 
    $\tilde{\Theta}\left(k\right)$, uniformly in $n$ and $k$.
\end{restatable}

Finally, when $k\geq n^{1/3}$,  we propose a different algorithm which exploits the small window regime. See Section~\ref{sec:profile-mapping-full-window} for the details.

\begin{restatable}{theorem}{thmrmmainrandommappingii}
    \label{thm:random-mapping-2}
    Algorithm~\ref{fig:procedure-final-mappings} produces the profile of a uniform random mapping with expected cost 
    $O\left((\log k)\cdot (\sqrt{\tfrac{n}{k}\log k} + \log k) \right)$, uniformly on $n\geq 1$, $k\geq 2$.
\end{restatable}

Combining these two algorithms, we obtain a sampler for the profile of a uniform random mapping that is optimal up to logarithmic factors over the whole parameter space. In particular, this answers an open question of Devroye and Los~\cite{devroyelos25}, discussed in the \hyperref[par:related-work]{related work paragraph} below.

\paragraph*{Random surjections.} 

A first important remark is that, when $k$ is \emph{small} compared to $n$, most random mappings are surjections. Indeed, classical results on the coupon collector problem show that, for $k\leq \tfrac{n}{\log n}$, a uniform random mapping from $[n]$ to $[k]$ is surjective \emph{with high probability} (whp for short) as $n\to \infty$, uniformly in $k$, as recalled in the following lemma.

\begin{restatable}[Coupon Collector]{lemma}{lemmacouponcollector}
\label{lemma:coupon-collector}
Let $p(n,k)$ be the probability of a random mapping from $[n]$ to $[k]$ being a surjection. For $k\leq \tfrac{n}{\log n}$ we have $p(n,k) \geq 1 - \tfrac{1}{\log n}$.
\end{restatable}
\begin{proof}
    The probability of a given value from $[k]$ not being in the image is $(1-\tfrac{1}{k})^n$. By the union bound  $1-p(n,k)\leq k\, (1-\tfrac{1}{k})^n$, which is smaller than $k\, \exp(-n/k)$ by $1+x\leq \exp(x)$, valid for all $x$. Using $k\leq \tfrac{n}{\log n}$ , we bound this by $ k\exp(-\log n)=\tfrac{k}{n}\leq \tfrac{1}{\log n}$ and we are done.
\end{proof}

Thus, when $k \leq \frac{n}{\log n}$, sampling profiles of surjections reduces to sampling profiles of random mappings. Moreover, as\reviewer{See Prop.~\ref{prop:length-surjections-from-mappings}.} $p(n,k)\to 1$ uniformly in $k \leq \frac{n}{\log n}$ as $n\to \infty$, we can show that the bound on the expected length of a profile presented in \autoref{prop:length-profiles-random} transfers verbatim to surjections.

Let us now consider the case\footnote{In this region of the parameter space, random mappings are not surjective whp. Indeed, the estimates in~\cite{erdHos1961classical} imply that if $k\geq \tfrac{n}{\log n} \times \big(1+2\tfrac{\log \log n}{\log n}\big)$, then $p(n,k)\to 0$.} $k\geq \frac{n}{\log n}$. In Section~\ref{section:general-algo-surjections}, we present an algorithm whose expected runtime is $\tilde{O}(1)$ for this region and is therefore optimal up to polylogarithmic factors. 

\begin{restatable}{theorem}{theoremalgosurjections}
        \label{thm:sampling-surjections}
    There is an algorithm\reviewer{The algorithm is described in Section~\ref{section:general-algo-surjections}. See Appendix~\ref{section:complete-pseudocode-surjections} for the pseudocode.} for sampling the profile of a uniform surjection with cost $O((\tfrac{n}{k})^2 + \tfrac{n}{k} \log k)$, uniformly on $n>k\geq n^{2/3}$. In particular, if $k\geq \tfrac{n}{\log n}$, this is $O((\log n)^2)$.
\end{restatable}

 \begin{wrapfigure}{R}{0.32\textwidth}
    \begin{minipage}{0.32\textwidth}
    \centering
    \begin{tikzpicture}[scale=0.5]
\begin{axis}[
    ybar interval,
    xlabel={Value},
    ylabel={Frequency},
]
\addplot[fill=green,
    fill opacity=0.2,
    draw opacity=0.6] coordinates {
(1, 40882256134781165)
(2, 65351495325307011)
(3, 69644198715747328)
(4, 55664155504689232)
(5, 35592319764731766)
(6, 18965125108196932)
(7, 8661803636542936)
(8, 3461537419020161)
(9, 1229637811754805)
(10, 393122466351600)
(11, 114257796196740)
(12, 30440742420775)
(13, 7486221382253)
(14, 1709563383828)
};
\end{axis}
\end{tikzpicture}
    \caption{Histogram corresponding to a random surjection with $n=10^{18},k = 3\times 10^{17}$. }
    \label{fig:placeholder}
    \end{minipage}
  \end{wrapfigure}
 As illustrated in \autoref{fig:placeholder}, the expected size of the profile of a random surjection is $\tilde{O}(1)$. This follows at once from the complexity of the algorithm.

\paragraph*{Methods.}
\label{par:methods} 

Our algorithms and results are based on exact-size Boltzmann samplers~\cite{boltzmann-flajolet-duchon-et-al} for the size-vectors of random mappings and random surjections. More precisely, random mappings from $[n]$ to $[k]$ are combinatorially specified (in the sense of \cite{FS}) by the labeled class ${\tt Seq}_{=k}({\tt Set}(\mathcal{Z}))$, which decomposes a mapping into the sequence of its $k$ preimages. The associated Boltzmann sampler first samples the size-vector before assigning the labels. To sample the size-vector, it draws $k$ iid Poisson random variables $X_1,\ldots,X_k$ with parameter $\omega=\omega(n,k)$ and accepts the tuple if its sum equals $n$. The parameter $\omega$ is chosen to maximize the acceptance probability by ensuring that $\E[X_1+\cdots+X_k]=n$, which gives $\omega=n/k$.
For surjections, the preimages must be nonempty, so the corresponding labeled class is ${\tt Seq}_{=k}({\tt Set}_{\neq\emptyset}(\mathcal{Z}))$. The sampler therefore draws $k$ iid Poisson random variables constrained to be at least one and again accepts the tuple if their sum equals $n$. As before, their parameter $\omega=\omega(n,k)$ is chosen so that $\E[X_1+\cdots+X_k]=n$.

Both rejection samplers can be analyzed in the same framework. At each round $j$, we sample the size-vector $(X_1^{\langle j\rangle},\ldots,X_k^{\langle j\rangle})$ and accept if and only if $X_1^{\langle j\rangle}+\ldots+X_k^{\langle j\rangle}=n$. Let $C_j=C_j(n,k)$ be the cost of the sampling round $j$ and let $\tau$ be the random index of the sample we accept for the first time. Observe that $C_1,C_2,\ldots$ are identically distributed and non-negative. Then Wald's Equation~\cite{Blackwell} describes the expected global cost\reviewer{See Annex~\ref{annex:walds-equation} for the details.}
\begin{equation}
    \label{eq:waldseq}
    \E[C_1+\ldots+C_\tau] = \E[C_1]\times \E[\tau]\,,\qquad \E[\tau]=1/a(n,k)\,,
\end{equation}
where $a(n,k)$ is the acceptance probability of the rejection algorithm.

\medskip
We propose an improvement of each of these factors, $\E[C_1]$ and $\E[\tau]$, as follows:
\begin{enumerate}
    \item {\em The expected cost of a single sampling round, factor $\E[C_1]$:} To sample the profile of a vector of $k$ iid random variables, Devroye~\cite{devroye1988multinomial,devroye2012gw} introduced the \emph{multinomial method}. We adapt this method to each of our cases, giving a precise analysis of the costs. This lowers the cost of each sampling round $\E[C_1]$.

    \smallskip
    \item {\em The expected number of sampling rounds, factor $\E[\tau]$:} It remains to consider the expected number of rounds $\E[\tau]=1/a(n,k)$. This factor can be as large as $\Theta(\sqrt{n})$ for our Boltzmann samplers (see Sections~\ref{sec:boltzmann-random-mappings} and~\ref{sec:boltzmann-random-surjections}). We bring this down significantly to at most a polylogarithmic number of rounds. To do this, we apply a method, which we think could rightfully be described as an \emph{anticipated acceptance} method. The idea traces its roots back to von Neumann~\cite{VonNeumann1951}, in order to simulate involved distributions, but it has more recently been applied by Arratia and DeSalvo~\cite{arratia2015probabilisticdivideandconquernewexact,ARRATIA_DeSALVO_2016,DESALVO201765}, and Bassino and Sportiello~\cite{BassinoSportiello2018LineartimeES}, in the context of discrete combinatorial structures.
    
    In short, the principle is the following: if the algorithm for rejection sampling involves, at one stage, a Bernoulli random variable with an acceptance probability that is ``uniformly small'', then we can renormalize it without altering the resulting distribution, thus increasing the acceptance probability. 
    This requires organizing the rejection step around conditional acceptance probabilities for which sufficiently sharp, efficiently computable upper bounds can be established.  Our main technical contribution is to make this acceptance rescaling effective uniformly in the two parameters $n$ and $k$. We develop different constructions for mappings and surjections and combine them with efficient candidate-profile generation to obtain the stated complexity guarantees.

\end{enumerate}

\paragraph*{Related work.}
\label{par:related-work}

Devroye and Los~\cite{devroyelos25} studied the efficient sampling of profiles of random mappings\footnote{In fact, Devroye and Los sample a vector which gives the number of preimages of each size from $0$ to the maximum size of a preimage. This vector contains the same information as the profile but can be larger due to the presence of the zero entries.} in the region $k \leq  n \leq k \log k$. This study is motivated by the efficient simulation of load-balancing algorithms. In this region, Devroye and Los~\cite{devroyelos25} give an algorithm in
$O\!\left(\nicefrac{\log k}{\log\left(\frac{4k\log k}{n}\right)}\right)$. Their algorithm matches the known lower bound for the expected size of the largest preimage whereas our algorithm runs in $O((\log k)^2)$ in this region. They leave open the question of whether it is possible to sample the profile when $n>k\log k$ with expected running time $O(\sqrt{\tfrac{n}{k}\log k})$. Our algorithms match this bound up to a $\log k$ factor and outperform it when $k$ is sufficiently small compared with $n$.

As mentioned at the beginning of the introduction, our algorithm can be used to sample set partitions of $[n]$ into $k$ parts. To our knowledge, the sampling of set partitions has only been studied in a model where the number of parts is not fixed. The partitions are drawn uniformly at random among all possible partitions of $[n]$, see~\cite{bodinidurand25,DESALVO201765}. 

As mentioned previously, our algorithms for sampling profiles of random surjections yield a sampler for uniform random surjections with expected cost $O(n \log k)$ that matches the entropy lower bound uniformly across the whole parameter space. To our knowledge, the random sampling of surjections has only been previously studied under the assumption that  $k \sim \alpha n$ for some fixed $\alpha \in (0,1)$. The complexity of the exact-size Boltzmann sampler for surjections has been analyzed precisely under the assumption that $k \sim \alpha n$ for some fixed $\alpha \in (0,1)$ in \cite{bassino-nicaud07}. In this regime, their expected complexity is $O(n^{3/2})$ as they perform $O(\sqrt{n})$ sampling rounds on average, each costing $O(n)$ samples from simple distributions such as geometric and constrained Poisson.

Still under this linearity assumption, Bassino and Sportiello~\cite{BassinoSportiello2018LineartimeES} obtained a sampling algorithm with an expected complexity of $O(n)$ samples from standard random variables\footnote{Here the notion of complexity does not count the number of fair random bits used but rather the number of samples from standard distributions. As each of these samples may require the production of a logarithmic number of fair random bits, there is no contradiction with the entropy lower bound.}. Rather than generating a surjection directly, they reduce the problem to sampling a substructure, which they call the {\em backbone} of the surjection. They manage to lower the rejection rate of the associated Boltzmann sampler to make it constant. More precisely, the expected number of sampling rounds is constant but depends on the ratio $\alpha = k/n$. In particular, as $\alpha$ tends to $0$, the constant is $\sim \sqrt{\nicefrac{2\alpha}{e}}\exp(1/(2\alpha))$, which tends to infinity exponentially in $\tfrac{1}{\alpha}$ (see the end of Section~4.1 in \cite{BassinoSportiello2018LineartimeES}). In our experiments, this algorithm is not usable in practice to sample a surjection from $[n]$ to $[\frac{n}{20}]$ for $n$ of the order of $10^6$. 

\paragraph*{Plan of the article.} 

As they are based on the exact-size Boltzmann samplers, our algorithms need an efficient procedure to sample the profile of $k$ iid Poisson variables, possibly constrained to be nonzero. In Section~\ref{sec:profiles-iid-poisson}, we present efficient algorithms for these problems based on Devroye's multinomial method. In Section~\ref{sec:profiles-random-mappings} and Section~\ref{sec:profiles-random-surjections}, we present our algorithms for sampling the profiles of random mappings and surjections, respectively.

\section{Profiles of $k$ iid constrained and unconstrained random variables}
\label{sec:profiles-iid-poisson}
In this section, we study the profiles of $k$ iid Poisson random variables $X_1,\ldots,X_k$ with parameter $\lambda>0$, both unconstrained and conditioned to be nonzero. The general idea is to use the multinomial method to sample the profile efficiently.

Let us recall the classical multinomial method in the unconstrained case and show why its direct application does not yield an efficient sampler in our setting. In this method, we sample the number $m_0$ of the $k$ Poisson variables that have value $0$ by drawing a $\Binomial(k,p_0)$, where $p_0=\Pr(\Poiss{\lambda}=0)$. We then sample the number $m_1$ of the remaining $k-m_0$ variables that have value $1$ by drawing a $\Binomial(k-m_0,p_1)$, where $p_1=\Pr(\Poiss{\lambda}=1\mid\Poiss{\lambda}\geq1)$, and so on. We stop as soon as all the variables have been assigned a value, that is, when $m_0+\cdots+m_\ell=k$. The expected number of binomial samples is at least $\E[\max(X_1,\ldots,X_k)]$, which is at least $\lambda=\E[X_1]$.

In the extended RAM model, where binomial random variables can be sampled in $O(1)$ expected time, this leads to an expected running time of $\Omega(\lambda)$. In our applications, $\lambda$ can be very large, making this cost potentially much larger than the expected profile length.

Indeed, in Section~\ref{sec:poisson-profile-length}, we give explicit bounds showing that the expected length of the profile is $O(1+\sqrt{\lambda\log k}+\log k)$, uniformly in $\lambda>0$ and $k\geq1$. In Section~\ref{sec:sampling-fast}, we then show how to modify the multinomial method to start from the mode of the Poisson distribution, obtaining an expected running time that matches this bound on the expected profile length.

Finally, in Section~\ref{sec:poisson-profile-constrained}, we treat the case of Poisson random variables conditioned to be nonzero, restricting our study to the parameter regimes needed for our applications.

\subsection{Expected length of the profile of $k$ iid Poisson random variables}
\label{sec:poisson-profile-length}

In this subsection, we summarize the results concerning the profile of $k$ iid Poisson random variables $X_1,\ldots,X_k$ with a fixed parameter $\lambda>0$. We first give the following upper bound on the expected length of the profile.

\begin{proposition}
    \label{prop:expected-length-iid-poisson}
    The length $L_{k,\lambda}$ of the profile of $k$ iid Poisson random variables of parameter $\lambda$ satisfies $\E[L_{k,\lambda}] \leq 1+4 \sqrt{\lambda \log k} + 3 \log k$ for all $\lambda>0$ and $k\in {\mathbb{Z}_{\geq 1}}$.\reviewer{See Appendix~\ref{annex:bounds-iid-poisson} for the full proof.}
\end{proposition}
\begin{proof}[Proof sketch]
    Observe that the length of the profile is at most $1+\max(X_1,\ldots,X_k)-\min(X_1,\ldots,X_k)$.
    We follow the classical Cram\'er--Chernoff method (see the proof of Theorem~2.5 in~\cite{concentrationinequalities}) to obtain upper bounds on the expected maximum and similar lower bounds on the expected minimum. 
    If $\log k<\lambda$, we show that
       $ \E[\max(X_1,\ldots,X_k)-\min(X_1,\ldots,X_k)]
        \leq 4\sqrt{\lambda\log k}\,.$
   If $\log k\geq\lambda$, the expected maximum is logarithmic in $k$:
     $   \E[\max(X_1,\ldots,X_k)]\leq 3\log k\,.$
    Combining these two inequalities gives the announced upper bound.
\end{proof}

We need a lower bound on the length of the profile that holds not only in expectation but also with high probability. For this, we exploit the concentration of Poisson random variables around their mode $\lfloor\lambda\rfloor$. More precisely, a window of width $\Theta(\sqrt{\lambda})$ around the mode carries a constant proportion of the probability mass, and the probabilities of individual values within this window are comparable. This yields the following theorem.\reviewer{See Appendix~\ref{annex:almost-uniform-window} for the proof.}

\begin{restatable}{theorem}{theoremlowerboundpoissonk}
    \label{thm:optimal-length-profile}
    Consider $k\in \mathbb{Z}_{\geq 1}$ and $\lambda > 0$. Let $X_1,\ldots,X_k$ be iid $\Poiss{\lambda}$ random variables. Given any $d >0$, there is $c=c(d)>0$ such that, for $k \geq d\,\sqrt{\lambda}$, the length $L_{k,\lambda}$ of the profile satisfies $L_{k,\lambda} \geq  c \,\sqrt{\lambda}$, with high probability as $\lambda\to\infty$, where the convergence in probability is uniform in $k$ satisfying the condition.
\end{restatable}

\subsection{Sampling the profile of $k$ iid Poisson random variables}
\label{sec:sampling-fast}

  \begin{wrapfigure}{R}{0.5\textwidth}
    \begin{minipage}{0.5\textwidth}
      \begin{algorithm}[H] \scriptsize               
      \caption{The improved {\em SimulateProfile} procedure, which samples the profile of $k$ iid ${\tt Poiss}(\lambda)$ random variables $(X_1,\ldots,X_k)$. We suppose we can effectively compute the probabilities involved: $p$, $r_j$ and $\ell_j$. \label{fig:procedure-simulate-profile} } 

         {\bf function} SimulateProfile($k$,$\lambda$):\\

        profile := $\{\}$\;

        $p$ := $\Pr(\Poiss{\lambda}\geq \lfloor \lambda \rfloor)$\;
        $R \gets {\tt Bin}(k,p)$\;
        $L := k - R$\;
        $j$ := $\lfloor \lambda \rfloor$\;
        \While{$R>0$}{
        
             $r_j := \Pr(\Poiss{\lambda}=j \, |  \, \Poiss{\lambda}\geq j)$\;
            $m \leftarrow {\tt Bin}(R,r_j)$\;

             add $(j,m)$ to profile {\bf if} $m>0$\;

            $R := R - m$\;
            $j := j+1$\;
        }
        $j := \lfloor \lambda \rfloor - 1$\;
        \While{$L>0$}{
          $\ell_j := \Pr(\Poiss{\lambda}=j \, |  \, \Poiss{\lambda}\leq j)$\;
          $m \leftarrow {\tt Bin}(L,\ell_j)$\;
 
            add $(j,m)$ to profile {\bf if} $m>0$\;

            $L := L - m$\;
            $j := j-1$\;
        }
    
    \Return profile\;
      \end{algorithm}
    \end{minipage}
  \end{wrapfigure}
  
As remarked above, Poisson random variables are concentrated around their mode $\lfloor\lambda\rfloor$. We can therefore adapt the multinomial method to start the exploration at $\lfloor\lambda\rfloor$ rather than at $0$.
First we decide, sampling a single binomial, how many Poisson random variables are greater than or equal to $\lfloor\lambda\rfloor$, and how many are strictly smaller. Then we iterate on each side using the multinomial method. The resulting algorithm is given in Algorithm~\ref{fig:procedure-simulate-profile}. The following proposition gives bounds on its expected running time.

\begin{restatable}{proposition}{propwindowmapping}
    \label{prop:runtime-iid-poisson-profile}
    \label{prop:prop-window-2}
    For $\lambda>0$ and $k\geq2$, the expected running time of Algorithm~\ref{fig:procedure-simulate-profile} is $O(\E[W_{k,\lambda}])$, where $W_{k,\lambda}$ is the number of binomial samples and satisfies $\E[W_{k,\lambda}]\leq 2+4\log k$ if $\lambda\leq\log k$ and $\E[W_{k,\lambda}]\leq 3+5\sqrt{\lambda\log k}$ if $\lambda>\log k$.
\end{restatable}
\begin{proof}[Proof sketch]
    The probabilities $p$, $\ell_j$ and $r_j$ can be computed in $O(1)$ time in the extended RAM model. The cumulative distribution of a Poisson random variable, that is, $\Pr(\Poiss{\lambda}\geq k)$ or $\Pr(\Poiss{\lambda}\leq k)$, corresponds to incomplete $\Gamma$-functions. Since binomial samples also take $O(1)$ expected time, the expected running time is proportional to the expected number of binomial samples.

    Writing $X_1,\ldots,X_k$ for the sampled Poisson variables, the iteration on the right side employs exactly $\max(X_1,\ldots,X_k,\lfloor\lambda\rfloor-1)-(\lfloor\lambda\rfloor-1)$ binomial samples, and the iteration on the left side employs $\lfloor\lambda\rfloor-\min(X_1,\ldots,X_k,\lfloor\lambda\rfloor)$. Including the initial sample (for $R$ in the pseudocode),
    we have $W_{k,\lambda}=2+\max(X_1,\ldots,X_k,\lfloor\lambda\rfloor-1)-\min(X_1,\ldots,X_k,\lfloor\lambda\rfloor)$.
    
    Extending the analysis in the proof of Proposition~\ref{prop:expected-length-iid-poisson} yields the two announced bounds.
\end{proof}

\subsection{Profile of $k$ iid Poisson random variables conditioned to be nonzero}
\label{sec:poisson-profile-constrained}

For Poisson random variables conditioned to be nonzero, we can use the same technique as in the proof of Proposition~\ref{prop:expected-length-iid-poisson} to bound the expected length of the profile.

\begin{restatable}{proposition}{propprofkpoissone}
    \label{prop:profile-k-poiss-1}
    Let $\lambda>0$ and $k\in\mathbb{Z}_{\geq1}$. Let $X_1,\ldots,X_k$ be iid $\Poissone{\lambda}$ random variables. Then $\E[\max(X_1,\ldots,X_k)]\leq 1+\frac{\log k+\lambda}{\log 2}$.
\end{restatable}
\reviewer{See Appendix~\ref{appendix:length-surjections-large-k} for the proof.}

In our application to random surjections in Section~\ref{sec:profiles-random-surjections}, since we consider $k\geq \tfrac{n}{\log n}$, we use this result with $\lambda\leq \log n$. Since the multinomial method starting at $1$ uses exactly $\max(X_1,\ldots,X_k)$ binomial samples, Proposition~\ref{prop:profile-k-poiss-1} gives an expected running time of $O(\log n)$, which is sufficient for our purposes.

\section{Profiles of random mappings}
\label{sec:profiles-random-mappings}

Following the two regimes for the expected profile length announced in Proposition~\ref{prop:length-profiles-random}, we distinguish the \emph{$k$ small regime}, $k\leq n^{1/3}$, and the \emph{full window regime}, $k\geq n^{1/3}$. The algorithm for the $k$ small regime is fairly straightforward and is described in Section~\ref{sec:profile-mapping-k-small}. In Sections~\ref{sec:boltzmann-random-mappings}, \ref{sec:early-acceptance-mappings} and~\ref{sec:profile-mapping-full-window}, we develop and analyze Algorithm~\ref{fig:procedure-final-mappings}, whose expected running time matches the expected profile length up to logarithmic factors. The proof of Proposition~\ref{prop:length-profiles-random} is postponed to Section~\ref{sec:profile-mapping-length}, since it exploits the existence of Algorithm~\ref{fig:procedure-final-mappings} to transfer bounds for the profiles of iid Poisson random variables to the profiles of random mappings.

\subsection{The case of small $k$: region $k \leq n^{1/3}$}
\label{sec:profile-mapping-k-small}

In this region of the parameter space, the expected length of the profile is $\Theta(k)$. 
This means that we can afford to sample the full size-vector of length $k$ to sample the profile. The random size-vector $(Y_1,\ldots,Y_k)$  follows trivially the multinomial distribution ${\tt Multi}(n;\nicefrac 1 k,\ldots,\nicefrac 1 k)$.   
\begin{wrapfigure}{R}{0.5\textwidth}
    \begin{minipage}{0.5\textwidth}
\begin{algorithm}[H]
\DontPrintSemicolon
\caption{The algorithm for small $k$.\label{algo:k-small}}
\For{$i=1,2,\ldots,k-1$}{
$n_i\gets {\tt Bin}(n-\sum_{j<i}n_i,\nicefrac{1}{(k-i+1)})$;
}
    $n_k\gets n-\sum_{i<k}n_i$; \;
\Return $profile(n_1,\ldots,n_k)$;
\end{algorithm}
\end{minipage}
\end{wrapfigure}

Using the Multinomial Method, this distribution can be sampled sequentially by drawing $k-1$ binomial variables as follows. First draw $Y_1$ as a binomial $n_1\gets{\Binomial}(n,\nicefrac{1}{k})$.  Then $Y_2$, conditioned on $Y_1=n_1$, is a binomial $n_2\gets {\Binomial}(n-n_1,\nicefrac{1}{(k-1)})$ and so on.

In our model, each sample of a binomial variable costs $O(1)$. The profile is then computed in $O(k \log k)$ by sorting the size vector. This leads to  the following.

\begin{proposition}
    For $k \leq n^{1/3}$, Algorithm~\ref{algo:k-small} samples the profiles of random mappings from $[n]$ to $[k]$ with expected complexity $\widetilde{O}(k)$.
\end{proposition}

\subsection{Exact-size Boltzmann sampler for random mappings}
\label{sec:boltzmann-random-mappings}

  \begin{wrapfigure}{R}{0.5\textwidth}
    \begin{minipage}{0.5\textwidth}
\begin{algorithm}[H]
\DontPrintSemicolon
\caption{Naive Boltzmann sampler for the size-vector of random mappings from $[n]$ to $[k]$. To obtain the profile, it suffices to count each of the distinct values.\label{algo:naive-boltzmann} }

\Repeat{$\sum_{j=1}^k X_j = n$;}{
\For{$i=1,2,\ldots,k$}{
$X_i\gets \Poiss{\omega := \frac{n}{k}}$;\;
}
}
\Return $(X_1,\ldots,X_k)$, the size-vector of the random mapping;\;
\end{algorithm}
\end{minipage}
\end{wrapfigure}
As explained in the \hyperref[par:methods]{Methods paragraph of the introduction}, the size-vector of a uniform random mapping from $[n]$ to $[k]$ can be sampled with an exact-size Boltzmann sampler (Algorithm~\ref{algo:naive-boltzmann}) as follows: sample $k$ iid random variables $X_1,\ldots,X_k$ with distribution $\Poiss{\omega}$, where $\omega=n/k$, and repeat until $\sum_{j=1}^k X_j=n$.

By Wald's equation~\cite{Blackwell}, the expected cost of this algorithm is $\frac{\mathbb{E}[C_{n,k}]}{a(n,k)}$, where $C_{n,k}$ is the cost of a single sampling round and $a(n,k)$ is the acceptance probability. The additive property of the Poisson distribution gives $a(n,k)=\Pr(\sum_{j=1}^k X_j=n)=\Pr(\Poiss{k\omega}=n)=\Pr(\Poiss{n}=n)\sim\frac{1}{\sqrt{2\pi n}}$. That is, $\Theta(\sqrt{n})$ sampling rounds are required.

Since we are only interested in the profile, we can directly sample the profile $\pi$ of these $k$ Poisson random variables using Algorithm~\ref{fig:procedure-simulate-profile} from Section~\ref{sec:sampling-fast}. We reject the profile unless the sum of its values, counted with their multiplicities, equals $n$, that is, unless $\sum_{(v,m)\in\pi}vm=n$.

 Following Proposition~\ref{prop:runtime-iid-poisson-profile}, applied with $\lambda=n/k$, we obtain a total expected cost of $O\!\left(\sqrt{n}\cdot \left(\sqrt{\frac{n}{k}\log k}+\log k\right)\right)$, uniformly on $n\geq1$, $k\geq2$. We now present a method to reduce the factor $\sqrt{n}$ arising from the low acceptance probability.

\subsection{The early acceptance method}
\label{sec:early-acceptance-mappings}

In this section, we explain the early acceptance technique mentioned in the introduction and popularized by DeSalvo and his collaborators. We consider, as a ``toy example'', the case of  Algorithm~\ref{algo:naive-boltzmann} when $k=2$. 
In this case, the complexity of the algorithm is given by the expected number of sampling rounds $\Theta(\sqrt{n})$.

A sampling round consists of sampling two iid Poisson random variables with parameter $\omega=\nicefrac{n}{2}$ and is rejected when the sum of the two variables $X_1$ and $X_2$ is different from $n$. An equivalent procedure would be to sample $x\gets X_1$ and, instead of simulating $X_2$, accept with probability $p_x$, which is the probability that the second Poisson variable $X_2$ equals $n-x$.  

The key observation underlying the method is that this  probability $p_x = \Pr(\Poiss{\omega}=n-x)$ is bounded independently of $x$ by $q  = \Pr(\Poiss{\omega}=\lfloor \omega \rfloor) \sim \nicefrac{1}{\sqrt{2 \pi  \frac{n}{2} }}$. \reviewer{See Lemma~\ref{lemma:mode-poisson}.} This directly follows from the fact that $\lfloor \omega \rfloor$ is a mode of $\Poiss{\omega}$.

Following this observation, consider a slight variation of the above algorithm, which uses the acceptance probability $\nicefrac{p_x}{q}$ instead of $p_x$. As $p_x \leq q$ for all $x$, $\nicefrac{p_x}{q}$ is indeed a probability. 

The probability of this new procedure generating a pair $(x_1,x_2)$ satisfying $x_1+x_2=n$ is just the probability of generating this pair given that we accept, namely
 \[
 \frac{\Pr(\Poiss{\omega}=x_1)\times \tfrac{p_{x_1}}{q}}{\sum_x \Pr(\Poiss{\omega}=x)\times \tfrac{p_{x}}{q} } =  \frac{\Pr(\Poiss{\omega}=x_1)\times p_{x_1}}{\sum_x \Pr(\Poiss{\omega}=x)\times p_x } \,.
 \]
 Note that the factor $q$ cancels out and we obtain the same probability as with the first rejection procedure. However, the acceptance probability for the second procedure is $a'=\frac{a}{q}$ where $a$ is the acceptance probability of the first procedure. A direct computation shows that $a'=\frac{\Pr(\Poiss{n}=n)}{\Pr(\Poiss{\omega}=\lfloor\omega\rfloor)}\sim \frac{1/\sqrt{2\pi n}}{1/\sqrt{2\pi \omega}} = \frac{\sqrt{2}}{2}$ as $\omega = \frac{n}{2}$.
 Hence the acceptance rate of the second procedure is constant, and the expected number of sampling rounds is $O(1)$.

\subsection{Optimal sampler for random mappings}
\label{sec:profile-mapping-full-window}

  \begin{wrapfigure}{R}{0.5\textwidth}
    \begin{minipage}{0.5\textwidth}

      \begin{algorithm}[H] \scriptsize               
      \caption{Algorithm to produce a random mapping profile. Here $D(x,a)=x^a e^{-x} / a!$.} 
\label{fig:procedure-final-mappings}
         function RandomMapProfile($n$,$k$):\\
    \If{$k$ is too small}{
        call procedure from Section~\ref{sec:profile-mapping-k-small}.
    }
    $\omega$ := $n/k$\;
    \While{True}{
        $k_\ell$ := $\lfloor k /2\rfloor$\;
        $k_r$ := $k-k_\ell$\;
        $profile_\ell\gets$ SimulateProfile($k_\ell$,$\omega$);\qquad\qquad \CommentSty{$\blacktriangleright$ Multinomial method} \;
        $s := sum(profile_\ell)$\; 
        \If{$s \leq n$}{
            $p $ := $D(\omega\cdot k_r,n-s)$ \;
            $q$  := $D(\omega\cdot k_r, \lfloor \omega \cdot k_r\rfloor)$\;
            \If{Bernoulli(p/q)}{
                $profile_r\gets$ RandomMapProfile($n-s$,$k_r$)\;
                \Return $(profile_\ell+profile_r)$\;
            }
         }
    }
    \
      \end{algorithm}
          \end{minipage}
  \end{wrapfigure}

Our general profile sampler (Algorithm~\ref{fig:procedure-final-mappings}) is recursive and produces the desired distribution for all values of $n$ and $k$.

If $k \leq n^{1/3}$, we apply the procedure described in Section~\ref{sec:profile-mapping-k-small}. Otherwise, we accelerate the sampling of the profile of $X_1,\ldots,X_k$, iid $\Poiss{\omega}$ random variables with $\omega=n/k$ conditioned on $X_1 + \cdots + X_k = n$. We sample the profile $\pi_\ell$ of the first $k_\ell=\lfloor k/2\rfloor$ Poisson variables using Algorithm~\ref{fig:procedure-simulate-profile}. From this profile, we compute the sum  $s = X_1 + \cdots + X_{k_\ell}$ . The probability $p_s$ that when sampling the $k_r = k - k_\ell$ remaining variables the total sum is $n$ is precisely $\Pr(X_{k_\ell+1} + \cdots + X_k = n - s) = \Pr(\Poiss{k_r\omega}=n-s)$. Remark that the probability $p_s$ is bounded uniformly in $s$ by the mode $q=\Pr(\Poiss{k_r\omega}=\lfloor k_r\omega\rfloor)$. Thus we again up-scale the acceptance rate by accepting with probability $\frac{\Pr(\Poiss{k_r\omega}=n-s)}{q}$. 

As in our toy example, the acceptance probability at this point is the acceptance probability of the naive algorithm $\Pr(\Poiss{n}=n)$ divided by $q$, which is equal to $\Pr(\Poiss{\frac{n}{2}}=\frac{n}{2})$ when $k$ is even. \reviewer{See Proposition~\ref{prop:up-scaled-proba-mappings}.} This quantity is bounded below by $\frac{\sqrt{2}}{2}\times (1-\tfrac{2.5}{n})\geq 0.1$ for all $n\geq 3$ and $k\geq 1$.

Once accepted, the profile of the first half of the preimage sizes is fixed to be $\pi_r$ and it remains to draw the profile of the second half, which is equivalent to uniformly sampling the profile $\pi_r$ of a random mapping from $[n-s]$ to $[k_r]$. This is then done recursively and we return the combined profile $\pi_\ell + \pi_r$. 
\reviewer{See the detailed proof in Section~\ref{sec:proof-main-rm}.}

\begin{restatable}{theorem}{thmrmmain}
    \label{thm:random-mapping}
    The global cost of Algorithm~\ref{fig:procedure-final-mappings} is 
    $O((\log k)(\sqrt{\tfrac{n}{k}\log k} + \log k))$, uniformly on $n\geq 1$ and $k\geq 2$.

\end{restatable}
\begin{proof}[Proof sketch.]
Let $(N_0,k_0), \ldots, (N_m,k_m)$ denote the values of the  parameters at the start of each recursive call. The pair $(N_0,k_0)$ is equal to $(n,k)$, and the values of the $k_i$'s are deterministic in $k$ and $i$ but the $N_i$'s are random variables. Observe that $m=O(\log k)$.

Let $S_{N_i,k_i}$ denote the total cost of the calls to $\texttt{SimulateProfile}$ during the $i$-th recursive call. For \emph{deterministic} $n$ and $k$, Wald's equation gives us the following upper bound $\mathbb{E}[S_{n,k}] \leq c(\sqrt{\tfrac{n}{k}\log k} + \log k)$ for some constant $c >0$ using Prop.~\ref{prop:prop-window-2} and the constant lower bound on the acceptance probability.

This, together with the concavity of $x\mapsto \sqrt{x}$, shows that, for now our random $N_i$, the inequality $\mathbb{E}[S_{N_i,k_i}] \leq c\, \left(\sqrt{\mathbb{E}[\nicefrac{N_i}{k_i}]\log k} + \log k\right)$ holds. To conclude, we establish that $\mathbb{E}[\nicefrac{N_i}{k_i}]=\nicefrac{n}{k}$ for each $i$. It follows that the expected total sampling cost of the algorithm is $\E[S_{N_0,k_0} + \cdots + S_{N_m,k_m}] \leq c'\, ( \log k)\, (\sqrt{\tfrac{n}{k}\log k} + \log k)$ for some constant~$c'$. 

With a suitable representation of the profiles, the complexity of summing the profiles is of the same order.
\end{proof}

\subsection{Expected length of the profile of a random mapping}
\label{sec:profile-mapping-length}

In this section, we give bounds on the expected length of the profile of a random mapping. These bounds are transferred from those obtained for the profiles of iid Poisson random variables with parameter $n/k$. 

We remark that Algorithm~\ref{fig:procedure-final-mappings} can be viewed as a rejection sampler for the half-profile of a random mapping from the profiles of $\lfloor k/2\rfloor$ iid $\Poiss{n/k}$ random variables. The key observation is that its acceptance probability is uniformly bounded below by a constant $\delta>0$. The existence of such an algorithm allows us to transfer upper bounds in expectation, as well as to obtain lower bounds in expectation from lower bounds holding with high probability.

\reviewer{See Appendix~\ref{sec:length-profile} for the full proof.}
\proplengthprofilesrandommapping*

\section{Profiles of surjections from $[n]$ to $[k]$ for $k \geq \frac{n}{\log n}$}
\label{sec:profiles-random-surjections}

As stated in the introduction, when $k\leq \tfrac{n}{\log n}$, a uniform random mapping from $[n]$ to $[k]$ is surjective with high probability. Hence, sampling the profile of a uniform random surjection reduces to sampling profiles of random mappings. In this section, we focus on the case $k\geq \tfrac{n}{\log n}$.

Our starting point is the exact-size Boltzmann sampler for the size-vector of a random surjection, described in Section~\ref{sec:boltzmann-random-surjections}, where we analyze its acceptance probability in detail. This analysis is more involved than for random mappings: the optimized parameter $\omega=\omega(n,k)$ is defined implicitly, and our estimates must hold uniformly in $n$ and $k$.

Once this estimate is established, Section~\ref{sec:upscaling-surjections-first} presents the general idea for up-scaling the acceptance probability and a first algorithm. For fixed $\omega$, this algorithm uses a constant expected number of sampling rounds, but this number grows exponentially as $\omega\to\infty$. Finally, Section~\ref{section:general-algo-surjections} gives an algorithm using $O(n/k)$ expected sampling rounds and having total expected runtime $\widetilde{O}(1)$, which is optimal up to logarithmic factors.

\subsection{Exact-size Boltzmann sampler}
\label{sec:boltzmann-random-surjections}

Following the case of random mappings, our starting point is the  Boltzmann sampler for size-vectors of random surjections. This algorithm samples $k$ iid $\Poissone{\omega}$ random variables $X_1, \ldots, X_k$ until $X_1+\cdots+X_k =n$. Here, the Poisson random variables are constrained to be non-zero: $\Pr(\Poissone{\omega} = j) = \tfrac{\omega^j}{j!}\,(e^\omega-1)^{-1}$ for $j\in \mathbb{Z}_{\geq 1}$.

The algorithm is correct for every choice of $\omega>0$, but the performance varies. To maximize the chance that $\sum_{j=1}^k X_j = n$, we choose $\omega=\omega(n,k)$ that makes $\sum_{j=1}^k \E[X_j] = n$, i.e., 
$k \frac{\omega e^\omega}{e^\omega -1} = n\,, $ known as the saddle-point equation. 
The solution is unique and depends only on $\nicefrac{n}{k}$, since $\omega \mapsto \frac{\omega e^\omega}{e^\omega -1}$ is strictly increasing. In fact, we can prove that $\omega=\xi+W_0(-\xi e^{-\xi})$, where $\xi:=n/k$ and $W_0$ is the $0$-th branch of the Lambert-W function~\cite{corless1996lambert}.

The following proposition summarizes the main features of the behaviour of $\omega(n,k)$. \reviewer{The proof can be found in Appendix~\ref{appendix:properties-omega-surjections}.}
\begin{restatable}{proposition}{propertiesomega}
        \label{prop:properties-omega-surjections}
    The solution $\omega=\omega(n,k) > 0$ of $k \frac{\omega e^\omega}{e^\omega -1} = n$ satisfies, for all $n>k\geq 1$,
    \[
    0 < \tfrac{n}{k}-1\leq \omega \leq \tfrac{n}{k}\,,\qquad \Delta \leq k\omega\leq 2 \Delta\,,\quad (\Delta:=n-k)\,, \qquad 0\leq 2\Delta -k\omega \leq \tfrac{1}{6}k \omega^2\,.
    \]

\end{restatable}

Similarly to the random mapping case, we can sample directly the profile of the vector $(X_1,\ldots,X_k)$ using the Multinomial Method (see Sec.~\ref{sec:sampling-fast}). Sampling in order $j=1,2,\ldots$, the Multinomial Method costs $\max(X_1,\ldots,X_k)$ samples from a binomial variable.  Following Prop.~\ref{prop:profile-k-poiss-1}, when  $k \geq \frac{n}{\log n}$, the expected cost of a sampling round is $O(\log n)$, and we do not need to start the sampling from the mode of the distribution. This bound is tight when $k \sim \frac{n}{\log n}$. In contrast, if $n-k=O(n^\theta)$ for some fixed $\theta<1$, the cost is $O(1)$.

\paragraph*{Uniform acceptance rate of the Boltzmann sampler.}

It remains to analyze the acceptance probability of the Boltzmann sampler. 
Unlike the random mapping case, the acceptance probability $a(n,k)$ of the exact-size Boltzmann sampler does not have a simple expression. However, we succeed in obtaining a uniform estimate of $a(n,k)$ that makes no assumption other than $n> k \geq 1$.

\begin{restatable}{theorem}{acceptanceprobasurjection}
        \label{thm:order-acceptance-prob-surj}
        Consider integers $1\leq k < n$.         
    The acceptance probability satisfies $a(n,k) = \Theta(\nicefrac{1}{\sqrt{\Delta}})$, with $\Delta=\Delta(n,k):=n-k$, where the hidden constants are  uniform. Equivalently, in terms of $k$ and $\omega$, $a(n,k)=\Theta(\nicefrac{1}{\sqrt{k\omega}})$, where the hidden constants are  uniform.
\end{restatable}

\begin{proof}[Proof sketch]
The acceptance probability $a(n,k)$ is equal to $[z^n] \left(\frac{e^{\omega z}-1}{e^\omega -1}\right)^k$, by considering the product of the $k$ probability generating functions. Then we notice that the estimate for $a(n,k)$ is equivalent to an estimate of $\stirlingsecond{n}{k}$. \reviewer{See Appendix~\ref{appendix:properties-omega-surjections} for the proof of Theorem~\ref{thm:order-acceptance-prob-surj}.} Indeed,  $\tfrac{k!}{n!}\stirlingsecond{n}{k} =   [z^n] (e^z-1)^k$; therefore, $\frac{\omega^n}{(e^\omega-1)^k} \tfrac{k!}{n!}  \stirlingsecond{n}{k}= [z^n] \left(\frac{e^{\omega z}-1}{e^\omega -1}\right)^k = a(n,k)$. Chelluri, Richmond and Temme have shown~\cite{ChelluriRichmondTemme} an estimate for $\stirlingsecond{n}{k}$, involving $\omega=\omega(n,k)$. We use Proposition~\ref{prop:properties-omega-surjections} to complete the proof.
\end{proof}
Thus the expected number of sampling rounds is $\Theta(\sqrt{n-k})$. Observe that this is always $O(\sqrt{n})$, but this latter estimate can be too pessimistic. For instance, if $n-k=O(n^\theta)$ for some $\theta \in [0,1)$, then the expected number of sampling rounds of the Boltzmann sampler is $O(n^{\theta/2})$.

\subsection{Up-scaling the acceptance probability for $k\geq \tfrac{n}{\log \log n}$}
\label{sec:upscaling-surjections-first}

To explain how we up-scale the acceptance probability, let us return to a single sampling round of the exact-size Boltzmann sampler. We need to sample the profile of $k$ iid $\Poissone{\omega}$ random variables, conditioned on their sum being $n$.
This amounts to sampling from ${\tt Multi}(k;p_1,p_2,\ldots)$, where $p_j=p_j(\omega):=\Pr(\Poissone{\omega}=j)$.

\smallskip
\noindent
\textbf{Poisson approximation.}  The main problem in applying an early acceptance method is that somewhere we need to \emph{unpack} from the multinomial some constrained Poisson random variables. In order to avoid the constrained Poisson, we follow the idea of classical Poisson (unconstrained!) approximation (see e.g.,~\cite[Section 5.4]{mitzenmacher}): to sample the multinomial we instead consider\footnote{We remark that this can be interpreted also as coming from a Boltzmann sampler in two variables. Arratia and DeSalvo discuss in~\cite[Section 3.3.2]{arratiadesalvo2022} a Boltzmann sampler, in two variables $\theta$ and $x$, optimized in order to target the desired size $n$ and number of components $k$ for an assembly (${\tt Set}$ construction in Analytic Combinatorics). In our case, instead of considering the whole space of two variables, we propose a \emph{two-step} procedure, which involves probabilities we can treat effectively. } independent Poisson random variables $Y_1=\Poiss{k\cdot p_1}, Y_2=\Poiss{k\cdot p_2},\ldots$. Subject to $\sum_{i=1}^\infty Y_i=k$, the vector $(Y_1,Y_2,\ldots)$ is distributed as ${\tt Multi}(k;p_1,p_2,\ldots)$. 

Adding such a step to our procedure introduces an additional acceptance condition, and the algorithm has an acceptance probability that is
\begin{center}
\scalebox{0.75}{
$\
a_2=a_2(n,k)= \underbrace{\Pr\Big(\sum_{i=1}^\infty Y_i=k\Big)}_{\Pr(\Poiss{k}=k)}\times \underbrace{\Pr\Big( \sum_{i=1}^\infty i Y_i=n \,\Big|\, \sum_{i=1}^\infty Y_i=k\Big)}_{a(n,k)}\,. $
}
\end{center}

\vspace{-.4cm}
The first factor, coming from the Poisson approximation, is $\Pr\Big(\sum_{i=1}^\infty Y_i=k\Big) = \Pr(\Poiss{k}=k)\sim \frac{1}{\sqrt{2\pi k}}$, while the second factor is simply $a=a(n,k)$, the acceptance probability of the original Boltzmann sampler. Thus $a_2=a_2(n,k)=\Theta( (k\cdot k\omega)^{-1/2})$ due to Theorem~\ref{thm:order-acceptance-prob-surj}. Notice that this acceptance probability is smaller than the previous $a$, as we have added a new condition.

\smallskip
\noindent
\textbf{Early acceptance.}  Now we show how to incorporate the early acceptance method into this new version of the algorithm.  With this scheme, the variables $Y_1,Y_2,\ldots$ are independent. Suppose that we first sample  $Y_3,Y_4,\ldots$ (more details on this later). Given their values and the two acceptance conditions, there is at most one possible value for the pair $(Y_1,Y_2)$ that meets the conditions. More precisely,  let $\hat{n}=n  - \sum_{j\geq 3} j Y_j$ and $\hat{k}=k-\sum_{j\geq 3} Y_j$. To satisfy the two acceptance conditions, it must be the case that $(Y_1,Y_2)=(2\hat{k}-\hat{n},\hat{n}-\hat{k})$. We can effectively compute the probability that this event occurs
$
p=\Pr(Y_1=2\hat{k}-\hat{n})\times \Pr(Y_2=\hat{n}-\hat{k} )\,.
$
The formula for $p$ then follows from the independence of $Y_1$ and $Y_2$. Instead of sampling $Y_1$ and $Y_2$, we introduce an acceptance Bernoulli test of parameter $p$ on which we will show that early acceptance can be applied. For this, we need a uniform bound on the random $p=p(Y_3,Y_4,\ldots)$.

The factors $\Pr(Y_1=j)$ and $\Pr(Y_2=j)$ can be bounded uniformly as $\max_t\Pr(\Poiss{\lambda}=t) \leq \Pr(\Poiss{\lambda}=\lfloor \lambda\rfloor)$. Thus let $q=\Pr(\Poiss{k\cdot p_1}=\lfloor k \cdot p_1\rfloor) \times \Pr(\Poiss{k\cdot p_2}=\lfloor k \cdot p_2\rfloor) $, which is our rescaling factor. We will prove in Corollary~\ref{corollary:algo-upper-upper} that this scaling factor leads to an expected number of sampling rounds of $O(\log n)$ overall.

  \begin{wrapfigure}{R}{0.53\textwidth}
    \begin{minipage}{0.53\textwidth}
      \begin{algorithm}[H] \scriptsize               
      \caption{\label{alg:surjection-upper-upper}Algorithm for a random surjection profile with an expected number of sampling rounds of $\Theta(\frac{e^\omega-1}{\omega})$ when $n>k\geq \tfrac{n}{\log n}$. The expected complexity is $O((\log n)^2)$ uniformly over $n>k\geq \tfrac{n}{\log \log n}$.
      } 
\label{fig:procedure-surjections-large-k}
         function RandomSurjectionProfile($n$,$k$):\\
    $\omega$ := SaddlePoint($n,k$)\;
    $p_1$ := $\nicefrac{\omega}{(e^\omega - 1)}$ \;
    $p_2$ := $\nicefrac{\tfrac{1}{2}\omega^2}{(e^\omega - 1)}$ \;
    $q$ := $ \Pr(\Poiss{k p_1}=\lfloor k p_1\rfloor) \times 
     \Pr(\Poiss{k p_2}=\lfloor k p_2\rfloor)$\;
    $p_r$ := $\Pr(\Poiss{\omega}\geq 3 |\Poiss{\omega}\geq 1)$\;
    \While{{\tt True}}{
        $profile$ := $\{\}$\;
        $R \gets {\tt Poiss}(k\cdot p_r)$\;
        $L := k - R$\;
        $N$ := $n$ \;
        $j$ := $3$\;
        \While{$R>0$}{
             $r_j := {\footnotesize\Pr(\Poiss{\omega}=j  |  \, \Poiss{\omega}\geq j)}$\;
            $m \leftarrow {\tt Bin}(R,r_j)$\;
             add $(j,m)$ to profile {\bf if} $m>0$\;
            $R := R - m$\;
            $N := N - m\cdot j$\;
            $j := j+1$\;
        }
        
        $m_1$ := $2L - N$\;
        $m_2$ := $N - L$ \;
        \If{$m_1\geq 0$ and $m_2\geq 0$}{
            $p $ := $\Pr(\Poiss{k p_1}=m_1)\times \Pr(\Poiss{kp_2}=m_2)$\;
            \If{Bernoulli(p/q)}{
                add $(1,m_1)$ to profile {\bf if} $m_1>0$\;
                add $(2,m_2)$ to profile {\bf if} $m_2>0$\;
                \Return $profile$\;
            }
         }
    }
    
      \end{algorithm}
          \end{minipage}
  \end{wrapfigure}

It remains to explain how to sample all of the variables $Y_3,Y_4,\ldots$. We do this by ``undoing'' the Poisson approximation: sample $R\gets \Poiss{k\cdot \sum_{j\geq 3}p_j}$, then, subject to $\sum_{j\geq 3} Y_j = R$, $(Y_3,Y_4,\ldots)$ is a multinomial ${\tt Multi}(R;\nicefrac{p_3}{\sum_{j\geq 3}p_j},\nicefrac{p_4}{\sum_{j\geq 3}p_j},\ldots)$.  The complete algorithm appears in Algorithm~\ref{alg:surjection-upper-upper}. The following proposition characterizes the expected number of sampling rounds of the algorithm.

\begin{proposition}
    \label{prop:sampling-round-upper-upper}
    Uniformly on  $n>k\geq \frac{n}{\log n}$, the expected number of sampling rounds in Algorithm~\ref{alg:surjection-upper-upper} is $\Theta(\frac{e^\omega-1}{\omega})$.
\end{proposition}
\begin{proof}[Proof sketch.] 
We begin by remarking that, for fixed $\delta>0$, $\Pr(\Poiss{\lambda}=\lfloor\lambda\rfloor)=\Theta(\lambda^{-1/2})$, uniformly on $\lambda\geq \delta$. Then we prove\reviewer{This is not immediate, see Appendix~\ref{appendix-proof-sampling-rounds-upper}.} that, for $k\geq \frac{n}{\log n}$, we have $kp_1\geq 1$ and $k p_2 \geq \frac{1}{4}$. Thus  $q=\Theta((k\sqrt{p_1p_2})^{-1})$ uniformly. Since $a_2=\Theta((k\sqrt{\omega})^{-1})$, the scaled acceptance probability is $a'_2=\Theta(\sqrt{\frac{p_1 p_2}{\omega}})$ and the result follows from $p_1=\frac{\omega}{e^\omega-1}$, $p_2=\frac{\omega^2/2}{e^\omega-1}$.
\end{proof}

\begin{corollary}
    \label{corollary:algo-upper-upper}
    When $k\geq \frac{n}{\log \log n}$, the expected number of sampling rounds in Algorithm~\ref{alg:surjection-upper-upper} is $O(\log n)$, uniformly in $(n,k)$.
         Moreover, if $k\geq c \, n$, for some fixed $c>0$, the expected number of sampling rounds is $O(1)$.
\end{corollary}
\begin{proof}
    For the first assertion, $\omega \leq n/k\leq \log \log n$, thus $\tfrac{e^\omega-1}{\omega}\leq e^\omega \leq \log n$. The other is similar.
\end{proof}

Finally, the total number of constrained Poisson random variables simulated in each sampling round of Algorithm~\ref{alg:surjection-upper-upper} is actually randomized: this is $K=Y_1+Y_2+Y_3+\ldots$ which is distributed as $\Poiss{k}$. We bound the cost of computing a profile by Jensen's inequality:
\begin{lemma}
    \label{lemma:prop:profile-k-randomized-poiss-1}
    Let $X_1,X_2,\ldots$ be iid $\Poissone{\omega}$ random variables. Consider an independent $K\sim\Poiss{k}$. Then $\E[\max(X_1,\ldots,X_K)]=O(1+\log k+\omega)$.
\end{lemma}

 Thus the overall cost of Algorithm~\ref{alg:surjection-upper-upper} is $O( (\log n)^2) $ uniformly over $n>k\geq \tfrac{n}{\log \log n}$.

\subsection{General up-scaling of the acceptance probability}
\label{section:general-algo-surjections}

Even though Algorithm~\ref{alg:surjection-upper-upper} involves a constant number of sampling rounds for fixed $\omega$, i.e., fixed $n/k$, its efficiency quickly deteriorates as $\omega\to \infty$ due to the exponential. In fact, for any fixed $\theta\in (0,1)$, when $\tfrac{n}{\log n}\leq k\leq \tfrac{n}{1+\theta \log n}$ the expected number of rounds is $\Omega(n^\theta /\log n)$. This is a consequence of the inequalities in Prop.~\ref{prop:properties-omega-surjections}.

This problem is solved by leaving $Y_a$ and $Y_{a+1}$ rather than $Y_1$ and $Y_2$ to the end, where $a:=\max(1,\lfloor \omega\rfloor)$. Notice that these are the Poisson random variables with the largest parameters, and, hence, the smallest probability at the mode. This maximizes the gains of the early acceptance method employed. The full pseudocode is provided in the appendix\reviewer{The pseudocode can be found in Appendix~\ref{section:complete-pseudocode-surjections}.}, as it is similar to Alg.~\ref{alg:surjection-upper-upper}. The resulting algorithm involves only $O(\omega)=O(n/k)$ rounds.

\begin{restatable}{theorem}{thmsamplingroundsalgouppercomplete}
    \label{thm:sampling-rounds-algo-upper-2}
    The expected number of sampling rounds is $\Theta(n/k)$, uniformly over $(n,k)$ satisfying $n>k\geq n^{2/3}$.
\end{restatable}
\begin{proof}[Proof sketch.]
When $\omega<2$,  $a=1$ is the mode and we obtain Algorithm~\ref{alg:surjection-upper-upper}, which requires $\Theta(1)$ when $\omega<2$. By Prop.~\ref{prop:properties-omega-surjections}, $1<\tfrac{n}{k}<3$ and the assertion holds. \reviewer{See Appendix~\ref{sec:proof-theorem-upper-complete} for the full proof.} 

When $\omega\geq 2$ and $k\geq n^{2/3}$, we apply twice the following property: fix $\delta>0$, then
{\small ${ 
    \Pr(\Poiss{\lambda}=\lfloor\lambda\rfloor)=\Theta(\lambda^{-1/2})
    }$} uniformly on $\lambda\geq \delta$. This resembles the proof of Prop.~\ref{prop:sampling-round-upper-upper}.
\end{proof}

Lemma~\ref{lemma:prop:profile-k-randomized-poiss-1} and Theorem~\ref{thm:sampling-rounds-algo-upper-2} together complete the proof of Theorem~\ref{thm:sampling-surjections} because $\omega\leq n/k$.

\begin{remark}
    The performance of the algorithm, although correct, degrades when $k=o(n^{2/3})$. 
     When $k=o(n^{2/3})$, the algorithm requires $\Theta(\sqrt{k n})$ sampling rounds. This case is not needed for our purposes: the case $k\leq \nicefrac
    {n}{\log n}$ is done by sampling random mappings.
\end{remark}

\section{Conclusion}
\label{sec:conclusion}
We presented algorithms for sampling profiles of random mappings and random surjections from $[n]$ to $[k]$, for all $n$ and $k$.
The expected complexity of these samplers in the extended RAM model matches the expected length of the output up to a logarithmic factor. Our complexity analysis does not assume any asymptotic relationship between $n$ and $k$. More precisely, our results cover all regions of interest for each case: all $n,k\in \mathbb{Z}_{\geq 1}$ for random mappings, and with the additional restriction $n>k$ for surjections.  

From a methodological point of view, we have shown how to combine two different techniques, the multinomial method on the one hand, and early acceptance on the other. We have provided a complete analysis for the composition of these techniques in the case of random mappings and surjections. This combination of techniques could probably be used in other situations, such as sampling the profiles (cycle structure) of permutations\footnote{And variations, like derangements (no fixed points), with restrictions on the possible lengths of cycles.} of $n$ with a fixed number of cycles $k$ when $k\geq \nicefrac{n}{\log n}$.

We designed our algorithms so that they can be implemented with ease without relying on arbitrary-precision libraries. The complexity in the calculation is concentrated solely in the black boxes we mentioned in the introduction. In particular, we need to be able to compute with care the probabilities and the cumulative probabilities for Poisson random variables. Table~\ref{fig:prob_special_functions} below shows the correspondence with special functions. These can be computed efficiently to high precision using known techniques~\cite{Temme_1994} implemented in standard libraries. 
  
{\small 
\begin{figure}[htbp]
    \centering
    \begin{tabular}{|c|c|c|}
        \hline
        \textbf{Probability} &
        \textbf{Special function} \\
        \hline
        $\Pr(\Poiss{\lambda}\geq k)$ & $P(k,\lambda)$, {regularized upper incomplete $\Gamma$-function}  \\
        $\Pr(\Poiss{\lambda}\leq k)$  & $Q(k+1,\lambda)$, {regularized lower incomplete $\Gamma$-function} \\
        $\Pr(\Poiss{\lambda} = k)$ & $\partial_\lambda P(k+1,\lambda)$, {derivative of regularized upper incomplete $\Gamma$-function} \\
        \hline
    \end{tabular}
        \caption{Probabilities and the corresponding special-function values.}
    \label{fig:prob_special_functions}
\end{figure}
}

\vspace{-.2cm}
Finally, one of the conclusions of the article is that the entropy of a random mapping or surjection stems mostly from the assignment of concrete labels to the underlying profile.
In particular, our results have direct implications for sampling random surjections. For $k\leq n/\log n$, uniform random surjections can be sampled in expected time $O(n\log k)$ in the word-RAM model by rejection from random mappings. For $k\geq n/\log n$, the same bound holds given the profile, as discussed in the introduction, simply by drawing two independent uniform permutations. Thus, if our $O((\log n)^2)$ profile sampler in the extended RAM model can be translated, as we expect, into a sampler with polylogarithmic expected running time in the word-RAM model, we would obtain an optimal sampler over the full parameter space, matching the entropy lower bound.

\bibliography{references}

\clearpage

\appendix

\setcounter{theorem}{0}
\renewcommand{\thetheorem}{A.\arabic{theorem}}
\renewcommand{\theproposition}{A.\arabic{theorem}}
\renewcommand{\thelemma}{A.\arabic{theorem}}
\renewcommand{\theremark}{A.\arabic{theorem}}

\section{General tools from Probability Theory}

\subsection{Wald's Equation}
\label{annex:walds-equation}
In this section we recall a central tool in our analysis: Wald's Equation. In our case, we need to adapt a classical variation of Wald's Equation that appears in~\cite{Blackwell}, but with one caveat: each sampling round $X_i$, seen as an sequence of random choices, is in itself a random variable, but the terms $C_i$ we sum in our case are a function of the result (the length of the profile, the length of the window, or, more simply, the number of random choices in the sampling round $X_i$), and the stopping time depends on the profiles (and so, on $X_i$) rather than the lengths. In the classical version of Wald's Equation, the event $\{\tau=n\}$ can be determined by $C_1,\ldots,C_n$, but here it is rather determined from the underlying random variables $X_i$.  Since $C_i=f(X_i)$, that is $C_i$ is a deterministic function of $X_i$,  the adaptation still holds.

This is just a direct adaptation of~\cite[Theorem 1]{Blackwell}. We provide a proof sketch for completeness, as this is not the typical direct application of Wald's Equation and we believe it is important to point out the difference and why this does not change the end result.

\begin{theorem}
\label{thm:wald}
Consider a probability space $(\Omega,\mathcal{F},\mu)$, where $\mathcal{F}$ is the $\sigma$-algebra of measurable subsets of $\Omega$.
Consider iid random variables $X_1,X_2,\ldots$ taking values in some measurable space $(X,\Sigma)$ and let $f\colon X\to \mathbb{R}$ be a measurable function, where we say $f(x)$ \emph{measures} of $x\in X$, such that $\E[|f(X)|]<\infty$.

Let $\tau\colon X\times X \times \ldots \to \mathbb{Z}_{\geq 1}$ satisfy that $S_k:=\{(x_1,x_2,\ldots,x_k,\ldots) : \tau(x_1,x_2,\ldots)=k\}$ depends only on the first $k$ entries $(x_1,\ldots,x_k)$ (therefore measurable), that $S_k$ are mutually disjoint and that $\sum_k \Pr(S_k)=1$, with the induced probability on $X^{\mathbb{N}}=X\times X \times \ldots$.

Then
\[
\E[f(X_1)+\ldots+f(X_\tau)] = \E[f(X_1)]\cdot \E[\tau]\,.
\]    

\end{theorem}
\begin{proof}[Proof sketch.]
This is a simple adaptation of the proof in~\cite[Thm.1]{Blackwell}. Without loss of generality we may assume that $\E[f(X_1)]=0$. By the strong law of large numbers we have
\begin{equation}
    \label{eq:largenumbers}
\frac{f(X_1)+\ldots+f(X_N)}{N} \to 0\,,
\end{equation}
with probability one. Notice that we have used the probability on $\mathbb{R}_{\geq 0}$ that is induced by $f$.

Define a sequence of random variables $n_1:=\tau(X_1,X_2,\ldots)$, $n_2=\tau(X_{n_1+1},X_{n_1+2},\ldots)$, $n_3=\tau(X_{n_1+n_2+1},\ldots)$ and so on, the successive stopping times. It is a standard check(verbatim as in~\cite{Blackwell}) that $n_1,n_2,\ldots$ are iid.

Define $W_k:=f(X_{n_1+\ldots+n_{k-1}+1})+\ldots+f(X_{n_1+\ldots+n_k})$, the sum of the \emph{values} between $n_{k-1}$ (exclusive) and $n_k$ (inclusive), where $n_0:=0$. These variables are also iid. We remark also that $W_k\geq 0$ and so $\E[W_k]$ is always well-defined (and might be infinite).

We claim that
\[
\frac{W_1+\ldots+W_k}{n_1+\ldots+n_k} \to 0\,,
\]
with probability one. Indeed, otherwise, using the definition of $W_i$ we would obtain a contradiction to Eq~\eqref{eq:largenumbers} over some (random) infinite subsequence of indices.

But then, with probability one
\[
\frac{W_1+\ldots+W_k}{k}=\frac{W_1+\ldots+W_k}{n_1+\ldots+n_k} \cdot \underbrace{\frac{n_1+\ldots+n_k}{k} }_{\to \E[n_1]=\E[\tau]} \to 0\,.
\]
Since $\frac{W_1+\ldots+W_k}{k} \to \E[W_1]$ with probability one, again by the strong law of large numbers, we conclude that this is $0$ and we are done.\end{proof}

\subsection{Lemmas for an anticipatedly accepted random variables}
\label{annex:lemmas-anticipated-accpetnace}

In this section we prove auxiliary results we require in order to derive bounds for the ``accepted random variable'' from the random variables we may discard. These are based on the fact that our anticipated acceptance procedures accept with probability that is bounded by below. 

\begin{lemma}
    \label{lemma-upper-bound-wald}
    Suppose a random variable $X_n \geq 0$ is thrown from a process of iid random variables $Y_1(n),Y_2(n),\ldots \geq 0$ by a rejection procedure with acceptance probability uniformly bounded from below by $\delta>0$. Then we have $\mathbb{E}[X_n] =O(\mathbb{E}[Y(n)])$ where $Y(n)$ is distributed as the $Y_i(n)$.
\end{lemma}
\begin{proof}
    Let $\tau = \tau(n) \geq 1$ be the stopping time associated with the acceptance, i.e., $X_n=Y_{\tau(n)}(n)$. Then Wald's identity implies
    $
    \mathbb{E}[\tau(n)] \times \mathbb{E}[Y(n)] = \mathbb{E}[\sum_{i\leq \tau} Y_i(n)]  \geq \mathbb{E}[X(n)]\,.
    $
    Observe that $\mathbb{E}[\tau(n)]=1/a(n)$, where $a(n)$ is the acceptance probability. By hypothesis $a(n)\geq \delta> 0$ so that $\mathbb{E}[\tau(n)]=1/a(n)\leq 1/\delta$. Therefore $\mathbb{E}[X(n)] = O(\mathbb{E}[Y(n)])$. 
\end{proof}

\begin{remark}
    If the acceptance probability $a(n)$ could not be bounded by below by $\delta>0$, we can still write $\E[X_n]=O(\E[Y(n)]/a(n))$ by the same argument.
\end{remark}

Lower-bounds require stronger conditions: we need an event that occurs with probability tending to one. This is needed in order to prevent the accepted set from being exceptional.

\begin{lemma}
    \label{lemma:transfer-lower-bound-ex}
    Suppose a random variable $X_n > 0$ is thrown from a process of iid random variables $Y_1(n),Y_2(n),\ldots > 0$ by a rejection procedure with acceptance probability uniformly bounded from below by $\delta>0$. Suppose that $Y_i(n)$ satisfy that $Y_i(n)\geq f(n) > 0$ with probability tending to $1$. Then we have $\mathbb{E}[X_n] = \Omega(f(n))$.
\end{lemma}
\begin{proof}
    Fix any $i\in \{1,2,\ldots\}$ and let $A_i=A_i(n)$ be the event that $Y_i(n)$ is accepted, while let $B_i=B_i(n)$ be the event that $Y_i(n)$ satisfies $Y_i(n)\geq f(n)$.

    By hypothesis $\Pr(A_i(n))\geq \delta$ while $\Pr(B_i(n))\to 1$ as $n\to \infty$. We show that this means that $\Pr(A_i(n)\cap B_i(n)) \geq \delta/2 $ for all large enough $n$. Indeed, otherwise  \[\Pr(B_i(n))=\Pr(A_i(n)\cap B_i(n)) + \Pr(B_i(n)\setminus A_i(n)) \leq \delta/2+ 1 - \Pr(A_i(n)) \leq 1 - \delta/2\] for arbitrarily large $n$, contradicting  $\Pr(B_i(n))\to 1$.

    Since we need a lower bound, it will actually be enough to look at $i=1$, i.e., $Y_1(n)$. Write $X_n \geq  Y_1(n) {\bf 1}_{A_1}$. Taking expected values and noticing that $\mathbb{E}[Y_1(n) {\bf 1}_{A_1}] \geq f(n) \Pr(A_1\cap B_1)$, we deduce $\mathbb{E}[X_n] \geq f(n) \delta/2$ for large enough $n$, completing the proof.
\end{proof}

\begin{remark}
    \label{remark:transfer-two-parameters}
    Lemma~\ref{lemma:transfer-lower-bound-ex} applies verbatim with two variables $n,k$ over a region where $\Pr(Y_i(n,k) \geq f(n) )\to 1$ uniformly on $k$ as $n\to\infty$, and the acceptance probability is at least $\delta>0$ uniformly. 
\end{remark}

The following is a complementary result in which we have just ``swapped'' which event happens with probability tending to one and which with probability at least $\delta>0$. The proof reads almost verbatim: the event that occurs with probability tending to one at some point intersects significantly with the one with probability at least $\delta>0$ for $n$ large enough.

\begin{lemma}
    Suppose a random variable $X_n > 0$ is thrown from a process of iid random variables $Y_1(n),Y_2(n),\ldots > 0$ by a rejection procedure with acceptance probability tending to one. Suppose that $Y_i(n)$ satisfy that $Y_i(n)\geq f(n) > 0$ with probability tending uniformly bounded by below by $\delta>0$. Then we have $\mathbb{E}[X_n] = \Omega(f(n))$.
\end{lemma}

\subsection{Cover the ``almost-uniform'' window argument}
\label{annex:almost-uniform-window}
In this section we present the proof of Theorem~\ref{thm:optimal-length-profile}. The argument is summarized by the following proposition: if there is a \emph{window} $W_\delta$,somewhere, that has positive probability at least $\delta>0$ and every element has roughly the same probability, then the size of the profile of $k$ iid random variables is $\Theta(k)$ whp up to at least $k=|W_\lambda|/2$.
\begin{proposition}
    Let $p_\lambda$ be a discrete on $\mathbb{Z}_{\geq 0}$ depending on a real parameter $\lambda>0$. Suppose that there were a \emph{window} $W_\lambda:=[a(\lambda),b(\lambda)]=\{a(\lambda),a(\lambda)+1,\ldots,b(\lambda)\}$ such that there exist constants $\delta>0$ and $c>0$, independent from $\lambda$, such that the following properties are satisfied
    \begin{enumerate}[a.]
        \item (Window is \emph{large}) $p(W_\lambda)\geq \delta >0$,
        \item (Elements in window are \emph{comparable}) $0<p_\lambda(=j) \leq c \,p_\lambda(=i)$ for all  $i,j\in W_\lambda$ .
    \end{enumerate}
    Consider iid random variables $X_1,\ldots,X_k$ distributed according to $p_\lambda$. Then
    \begin{enumerate}[i)]
        \item If $k\leq |W_\lambda|/2$, we have that the number of distinct elements $X_1,\ldots,X_k$, namely $\#\{X_1,\ldots,X_k\}$, is $\Theta(k)$ with exponentially high probability in $k$, where the constants in the $\Theta$ term and the probability are uniform in $\lambda$.
        \item Suppose further that $|W_\lambda|\to\infty$ as $\lambda\to \infty$. For any fixed $d>0$, over $k\geq d\times |W_\lambda| $ we have $\#\{X_1,\ldots,X_k\}=\Omega(|W_\lambda|)$ whp as $\lambda\to\infty$, where the constant of the $\Omega$ term and the convergence in probability are both uniform on $k$ satisfying the condition.
    \end{enumerate}
\end{proposition}
\begin{proof}
    Part $(ii)$ follows from $(i)$ as the number of elements is non-decreasing in $k$. Indeed, picking the smallest $k=k(\lambda):=\lceil d\times |W_\lambda| \rceil $, part $(i)$ implies that $c_1 k\leq \#\{X_1,\ldots,X_k\}$  for some constant $c_1$ independent from $\lambda$, with probability at least $1-\alpha e^{-\beta k}$ for some $\alpha,\beta>0$ also indep. from $\lambda$. Thus, for any $k'\geq k(\lambda)$ we have that $\#\{X_1,\ldots,X_{k'}\}\geq c_1 k(d) \geq c_1 d |W_\lambda|$ with probability tending to one as $\lambda\to \infty$, uniformly on $k'\geq k(\lambda)$. 
    
    We now prove $(i)$. Let $C_k$ be the number of elements from $W_\lambda$ seen at time $k$. Then
    $C_{k+1}-C_k$, given $C_k$, is $+1$ with probability at least $\delta (|W_\lambda|-C_k)/(c |W_\lambda|)$. Indeed, $\delta\leq p_\lambda(W_\lambda)=\sum_{j\in W_\lambda} p_\lambda(j) \leq c \sum_{j\in W_\lambda} p_\lambda(i) = c |W_\lambda| p_\lambda(i)$ for any $i\in W_\lambda$. As long as $k\leq |W_\lambda|/2$, we surely have $C_k \leq |W_\lambda|/2$ and $\delta (|W_\lambda|-C_k)/(c |W_\lambda|) \geq \delta /(2c)$. Thus we may bound $C_k$ by a binomial random variable ${\tt Bin}(k,\delta/(2c))$. By the Hoeffding bounds (see e.g.,~\cite[Thm. 2.8]{concentrationinequalities}), this is at least $\delta/(4c) k$ with exponentially high probability in $k\to \infty$, uniformly in $\lambda$. On the other hand, the upper bound $\#\{X_1,\ldots,X_k\}\leq k$ holds always.
\end{proof}

Note that this window concept applies immediately to Poisson random variables: let $W_\lambda:=[\lfloor\lambda\rfloor-\lfloor\sqrt{\lambda}\rfloor,\lfloor\lambda\rfloor+\lfloor\sqrt{\lambda}\rfloor]$. Then we can pick $\delta=1/9$ and $c=e$. The value $\delta=1/9$ follows, for example, directly from the Vysochanskij–Petunin inequality concerning unimodal random variables. Of course, much tighter values are possible for $\delta$. The value for $c$ is a calculation involving the inequalities related to Stirling's formula in~\cite{Rob55}.
As a consequence we obtain the following result by considering part $(ii)$ of the previous proposition:
\theoremlowerboundpoissonk*

\section{Inequalities for the maximun and minimum of $k$ iid Poisson random variables}
\label{annex:bounds-iid-poisson}

In this section, we derive bounds on the maximum and minimum of iid Poisson random variables using the following general lemma.

\begin{restatable}{lemma}{lemmaineqmaxmin}
    \label{lemma:inequality-max-exp-rv}
    Fix $k\geq 1$ and consider iid random variates $X_1,\ldots,X_k$ taking on values on $\mathbb{Z}_{\geq 0}$. Let $p_j = P(X_i=j)$ and $F(x)=\sum p_j x^j$ be their pgf. Assume $F(x)$ to be convergent on $J:=(1,\rho)$, where $\rho>1$ might be infinite.
    \begin{enumerate}
        \item[a)] For every $a\in J$,
    \begin{center}
    {\small
$    \mathbb{E}[\min(X_1,\ldots,X_k)]\geq \displaystyle\frac{-\log k + \log (1/F(a))}{\log (1/a)}\,.
$    }
\end{center}
        \item[b)] For every  $a\in (0,1)$,
    \begin{center}
    {\small
$    \mathbb{E}[\min(X_1,\ldots,X_k)]\geq \displaystyle\frac{-\log k + \log (1/F(a))}{\log (1/a)}\,.
$    }
\end{center}
    \end{enumerate}
\end{restatable}

\begin{proof}
    By Jensen's inequality $a^{\E[\max(X_1,\ldots,X_k)]} \leq \E[a^{\max(X_1,\ldots,X_k)}]$, for $a\in (0,1)$ or $a>1$. Then
\[
a^{\E[\max(X_1,\ldots,X_k)]} \leq \E[a^{\max(X_1,\ldots,X_k)}] \leq \E\Big[\sum_{i=1}^k a^{X_i}\Big] = k \E[a^X] \,.
\]
\vspace{-0.1cm}
As $\E[a^X] = F(a)$, the proof follows by taking logs. When $a>1$ we have that $\log a > 0$ and the sense of the inequalities are preserved. 

For the minimum, the same inequalities $a^{\mathbb{E}[\min(X_1,\ldots,X_k)]}\leq k \mathbb{E}[a^X]$ holds for all $a>0$. If $a<1$, then $\log a < 0$ and the sense of the inequalities is reversed.
\end{proof}

We derive concrete bounds for the maximum and minimum of iid Poisson random variables using Lemma~\ref{lemma:inequality-max-exp-rv} by optimizing over the parameter $a$.

\begin{corollary}
    \label{cor:max-k-poiss}
        Let $\lambda > 0$ and $k\in \mathbb{Z}_{\geq 1}$ such that  $\log k \geq \lambda$. Let $X_1,\ldots,X_k$ be iid Poisson random variables with parameter $\lambda$, then
        \[
        \mathbb{E}[\max(X_1,\ldots,X_k)] \leq 3 \log k\,.
        \]
\end{corollary}

\begin{proof}
        We employ Lemma~\ref{lemma:inequality-max-exp-rv}. Remark that $F(x) = e^{\lambda (x-1)}$ in this case, and that the hypothesis are clearly verified as the radius of convergence is infinite.

    Set $a=e$ in Lemma~\ref{lemma:inequality-max-exp-rv}. We deduce, as $\lambda \leq \log k$  by hypothesis
    \[
    \mathbb{E}[\max(X_1,\ldots,X_k)] \leq \log k + \lambda\times (e-1) \leq e\times \log k\,.
    \]
    The final inequality follows from $e<3$.
\end{proof}

\begin{restatable}{proposition}{propdiffmaxminkpoisson}
    \label{prop:diff-max-min-k-poisson}
    Let $\lambda > 0$ and $k\in \mathbb{Z}_{\geq 1}$ such that $\log k < \lambda$. Let $X_1,\ldots,X_k$ be iid Poisson random variables with parameter $\lambda$, then
    \[
        \mathbb{E}[\max(X_1,\ldots,X_k)-\min(X_1,\ldots,X_k)] \leq 4\sqrt{\lambda\log k}\,.
    \]
\end{restatable}
\begin{proof}
    We employ Lemma~\ref{lemma:inequality-max-exp-rv}. Remark that $F(x) = e^{\lambda (x-1)}$ in this case, and that the hypothesis are clearly verified as the radius of convergence is infinite.

    Let $a=1 + \sqrt{\tfrac{\log k}{\lambda}}$ on the item~(a) from Lemma~\ref{lemma:inequality-max-exp-rv}. We have
    \[
    \mathbb{E}[\max(X_1,\ldots,X_k)] \leq \frac{\log k + \lambda \sqrt{\tfrac{\log k}{\lambda}}}{ \log (1 + \sqrt{\tfrac{\log k}{\lambda}})} \,.
    \]

    We remark that, for $t\in [0,1)$, $\log(1+t)\geq t - \tfrac{t^2}{2}$. Thus, setting $t= \sqrt{\tfrac{\log k}{\lambda}}$, we deduce $\log (1 + \sqrt{\tfrac{\log k}{\lambda}}) \geq  \sqrt{\tfrac{\log k}{\lambda}} \times (1  - \tfrac{1}{2} \sqrt{\tfrac{\log k}{\lambda}})$. 
    
    Remark that, for $t\in (0,1)$, \[\tfrac{1}{1-\tfrac{1}{2}t} = 1 + \tfrac{1}{2}t+\tfrac{1}{4}t^2+\ldots \leq 1+ t (\tfrac{1}{2}+\tfrac{1}{4}+\ldots) = 1 + t\,. \]
    Thus we conclude, combining these last two inequalities, \[
    \tfrac{1}{\log(1+ \sqrt{\tfrac{\log k}{\lambda}})} \leq \frac{1}{ \sqrt{\tfrac{\log k}{\lambda}} \times (1 - \tfrac{1}{2} \sqrt{\tfrac{\log k}{\lambda}})} \leq \frac{1+ \sqrt{\tfrac{\log k}{\lambda}}}{ \sqrt{\tfrac{\log k}{\lambda}}}\,.
    \]
    We deduce
    \[
    \mathbb{E}[\max(X_1,\ldots,X_k)] \leq \left(\log k + \lambda \sqrt{\tfrac{\log k}{\lambda}}\right)\times \frac{1+ \sqrt{\tfrac{\log k}{\lambda}}}{ \sqrt{\tfrac{\log k}{\lambda}}}\ = \lambda + 2 \sqrt{\lambda \log k} + \log k\,.
    \]

    Next, set $a=1 - \sqrt{\tfrac{\log k}{\lambda}} > 0$ in item~(b) of Lemma~\ref{lemma:inequality-max-exp-rv},
        \[
    \mathbb{E}[\min(X_1,\ldots,X_k)] \geq \frac{-\log k + \lambda \sqrt{\tfrac{\log k}{\lambda}}}{ \log \left(\nicefrac{1}{(1- \sqrt{\frac{\log k}{\lambda}})}\right)} \,.
    \]
    We remark that the numerator is positive, indeed
    \[
    -\log k + \lambda \sqrt{\tfrac{\log k}{\lambda}} = -\log k + \sqrt{\lambda \log k} > 0\Leftrightarrow \sqrt{\lambda} > \sqrt{\log k} \Leftrightarrow \lambda > \log k\,.
    \]
    Thus we concentrate on the denominator. We recall that $\log \frac{1}{1-t}  = t + \tfrac{t^2}{2}+\tfrac{t^3}{3}+\ldots\leq t + t^2+t^3+\ldots = \frac{t}{1-t}$ for all $t\in (0,1)$. It follows that
    \[
    \frac{1}{ \log \left(\nicefrac{1}{(1- \sqrt{\frac{\log k}{\lambda}})}\right)} \geq \frac{1-\sqrt{\frac{\log k}{\lambda}}}{\sqrt{\frac{\log k}{\lambda}}}\,.
    \]
    We obtain, as the numerator is positive,
    \[
    \mathbb{E}[\min(X_1,\ldots,X_k)] \geq (-\log k + \lambda \sqrt{\tfrac{\log k}{\lambda}})\times \frac{1-\sqrt{\frac{\log k}{\lambda}}}{\sqrt{\frac{\log k}{\lambda}}} =  \lambda - 2 \sqrt{\lambda \log k} + \log k \,.
    \]
    Taking the differences between the inequalities we obtain the bound.
\end{proof}

Lastly, we the consider the case $k\gg n$ when we have $k$ iid Poisson random variables of parameter $\omega=n/k$. Consider $k = \Omega( n^{\theta})$ for some $\theta>1$, we show that $\E[\max(X_1,\ldots,X_k)]$ is bounded by a constant depending only on $\theta$.

\begin{proposition}
    Let $k \geq n$ be positive integers such that $k \geq n^\theta$ for some fixed $\theta > 1$. Consider iid random variables $X_1,\ldots,X_k$ distributed according to ${\tt Poiss}(\omega)$ where $\omega=\nicefrac{n}{k}$.

    Then we have \begin{equation}
        \mathbb{E}[\max(X_1,\ldots,X_k)] \leq \frac{1}{1-\tfrac{1}{\theta}} \times \left(1 + \tfrac{1}{\log k}\right)\,.
    \end{equation}
\end{proposition}
\begin{proof}
Pick $a = 1/\omega = k/n$ in item~(a) of Lemma~\ref{lemma:inequality-max-exp-rv}, then we obtain       $
    \mathbb{E}[\max(X_1,\ldots,X_k)] \leq  \frac{\log k + 1}{\log k - \log n} = \frac{1}{1-\tfrac{\log n}{\log k}} \times\left(1+ \tfrac{1}{\log k}\right)\,. $ Since $\log n / \log k  \leq 1/\theta$, we are done.
\end{proof}

\section{Inequalities for the maximum of $k$ iid  Poisson random variables constrained to be non-zero}

\label{appendix:length-surjections-large-k}

In this section we give simple bounds for the maximum of  iid Poisson random variables constrained to be non-zero using \autoref{lemma:inequality-max-exp-rv}.

\propprofkpoissone*
\begin{proof}
        We employ Lemma~\ref{lemma:inequality-max-exp-rv}. Remark that $F(x) = \frac{e^{\lambda x} - 1}{e^\lambda - 1}$ in this case, and that the hypothesis are clearly verified as the radius of convergence is infinite.

        Set $a=2$ in Lemma~\ref{lemma:inequality-max-exp-rv}. Then 
        \[
        F(2) = \frac{e^{2\lambda}-1}{e^\lambda -  1} = e^{\lambda}+1 \leq 2 e^\lambda\,.
        \]
        From the inequality we obtain
        \[
        \E[\max(X_1,\ldots,X_k)] \leq \frac{\log k + \lambda + \log 2}{\log 2} \,,
        \]
        as was to be proved.
\end{proof}

As an consequence we obtain the following proposition
\begin{restatable}{proposition}{profilehigh}
    \label{maximum-high-part}
    Let $n > k\geq 1$ be integers. Consider iid random variables $X_1,\ldots,X_k$ distributed according to ${\tt Poiss}_{\geq 1}(\omega)$ where $\omega=\omega(n,k)$ satisfies $k\omega \frac{e^\omega}{e^\omega - 1 }=n$.

    If $k\geq \frac{n}{\log n}$, we have $\E[\max(X_1,\ldots,X_k)]\leq 4 \log n\,.$
\end{restatable}
\begin{proof}
    We apply Proposition~\ref{prop:profile-k-poiss-1}. Notice that $1\leq \nicefrac{n}{k}\leq \log n$, $k\leq n$, and  $\omega\leq \frac{n}{k}\leq \log n$, thus $\E[\max(X_1,\ldots,X_k)]\leq (1 + \tfrac{2}{\log 2})\log n$. Then the inequality follows from $1 + \tfrac{2}{\log 2}\leq 4$.
\end{proof}

These results, however, are quite pessimistic for the case $\alpha=k/n\to 1$. 
The following result shows that the expected value is constant as soon as $m(n)=O(n^\theta)$ for $\theta<1$.
\begin{proposition}
    Let $n > k$ be positive integers such that $n - k = m(n)$ with $m(n)=O(n^\theta)$ for some fixed $\theta \in (0,1)$. Consider iid random variables $X_1,\ldots,X_k$ distributed according to ${\tt Poiss}_{\geq 1}(\omega)$ where $\omega=\omega(n,k)$ satisfies the saddle-point equation $k\omega \frac{e^\omega}{e^\omega - 1 }=n$.

    Then we have \begin{equation}
        \mathbb{E}[\max(X_1,\ldots,X_k)] \leq 2 + \tfrac{\theta}{1-\theta} + o_\theta(1)\,.
    \end{equation}
\end{proposition}
\begin{proof}
    In this case $\omega = \omega(n) \sim 2 m(n)/n \to 0$. Pick $a = 1/\omega(n)$ in item~(a) of Lemma~\ref{lemma:inequality-max-exp-rv}, then we obtain       $
    \mathbb{E}[\max(X_1,\ldots,X_k)] \leq  2 + \frac{\log m(n)}{\log n - \log m(n)} + o_\theta(1)\,.$
\end{proof}

\section{Expected cost of the profile sampler for random mapping}

In this section we study the global cost of Algorithm~\ref{fig:procedure-final-mappings}. First, Annex~\ref{annex:displacement-poisson} introduces some useful inequalities regarding Poisson random variables that are needed in order to bound the cost of simulating a profile, namely Algorithm~\ref{fig:procedure-simulate-profile}. This is done by proving Proposition~\ref{prop:prop-window-2}. We do this in Annex~\ref{sec:proof-window}. Finally, Annex~\ref{sec:proof-main-rm} explains how to integrate this into the study of the recursive procedure Algorithm~\ref{fig:procedure-final-mappings}.

\subsection{Mode and average displacement of Poisson random variables}
\label{annex:displacement-poisson}

\newcommand{\poisson}[2]{\mathrm{Poiss}_{#1}(#2)}

It is well-known to $\lfloor \lambda \rfloor$ is a mode for a Poisson random variable of parameter $\lambda$ (see for instance \cite[Eq.~4.22]{JohnsonKK2005}).

\begin{lemma}
    \label{lemma:mode-poisson}
For $\lambda > 0$ and for all $k \geq 0$, $\poisson{\lambda}{k} \leq \poisson{\lambda}{\lfloor \lambda \rfloor}$
\end{lemma}

We will need the following bounds on the maximum probability of a Poisson random variable.

\begin{lemma}
\label{lemma:max-probability-poisson}
For $\lambda\geq 1$, $\Pr(\Poiss{\lambda}=\lfloor\lambda\rfloor)=\frac{\lambda^{\lfloor\lambda\rfloor}}{(\lfloor\lambda\rfloor)!}e^{-\lambda}$ satisfies
    \[
    \frac{e^{-2}}{\sqrt{2\pi \lfloor\lambda\rfloor}} \leq \frac{\lambda^{\lfloor\lambda\rfloor}}{(\lfloor\lambda\rfloor)!}e^{-\lambda}\leq \frac{1}{\sqrt{2\pi \lfloor\lambda\rfloor}}
    \]
\end{lemma}
\begin{proof}
    Writing $\{\cdot\}$ for the fractional part $\{x\}:=x-\lfloor x\rfloor$, we have
    \[
    \frac{\lambda^{\lfloor\lambda\rfloor}}{(\lfloor\lambda\rfloor)!}e^{-\lambda}= \frac{\lfloor\lambda\rfloor^{\lfloor\lambda\rfloor}}{(\lfloor\lambda\rfloor)!}e^{-\lfloor\lambda\rfloor} \left( 1 + \frac{\{\lambda\}}{\lfloor\lambda\rfloor}\right)^{\lfloor\lambda\rfloor} e^{-\{\lambda\}}
    \]
     By the inequality $1+x\leq e^x$, valid for all $x$, we deduce
    \[
    \left( 1 + \frac{\{\lambda\}}{\lfloor\lambda\rfloor}\right)^{\lfloor\lambda\rfloor} e^{-\{\lambda\}}\leq 1\,.
    \]
    While, by the classical inequalities for the Stirling formula~\cite{Rob55} we have
    \[
    \frac{\lfloor\lambda\rfloor^{\lfloor\lambda\rfloor}}{(\lfloor\lambda\rfloor)!}e^{-\lfloor\lambda\rfloor} \leq \frac{1}{\sqrt{2\pi \lfloor\lambda\rfloor}}\,. 
    \]

    For the lower bound, it is clear that $ \frac{\lambda^{\lfloor\lambda\rfloor}}{(\lfloor\lambda\rfloor)!}e^{-\lambda} \geq  \frac{\lfloor\lambda\rfloor^{\lfloor\lambda\rfloor}}{(\lfloor\lambda\rfloor)!}e^{-(1+\lfloor\lambda\rfloor)}$
            By the classical inequalities for Stirling formula~\cite{Rob55} 
    \[
    \frac{\lfloor\lambda\rfloor^{\lfloor\lambda\rfloor}}{(\lfloor\lambda\rfloor)!}e^{-\lfloor\lambda\rfloor} \geq \frac{1}{\sqrt{2\pi \lfloor\lambda\rfloor}} e^{-1/(12(\lfloor\lambda\rfloor+1))} \geq \frac{e^{-1}}{\sqrt{2\pi \lfloor\lambda\rfloor}} \,,
    \]
    thus the proof is complete. \qedhere
\end{proof}

The deviation of a Poisson random variable $X$ with parameter $\lambda >0$ from its mean satisfies the following.

\begin{proposition}[{see \cite[Eq.~4.19]{JohnsonKK2005}}]
Let $\lambda \geq 1$ and let $X$ be a Poisson random variable with parameter $\lambda$. Then    \[
\mathbb{E}[|X-\lambda|] = 2 \frac{\lambda^{\lfloor\lambda\rfloor+1}}{(\lfloor\lambda\rfloor)!} e^{-\lambda}\,.
\]
\end{proposition}

Given the formula for the average displacement from the mean, we use Lemma~\ref{lemma:max-probability-poisson} to obtain bounds on the displacement.

\begin{proposition}
    \label{prop:deviation-mean-poisson}
    Let $\lambda \geq 1$ and let $X$ be a Poisson random variable with parameter $\lambda$. Then
    \[
    \frac{2}{e^2}\sqrt{\frac{2}{\pi}\lambda}\leq \mathbb{E}[|X-\lambda|] \leq \sqrt{\frac{2}{\pi}\lambda} + \sqrt{\frac{2}{\pi\lfloor\lambda\rfloor}} \,.
    \]
\end{proposition}

\begin{proof}
    By the inequality for the mode, see Lemma~\ref{lemma:max-probability-poisson}, we find
    \[
    \mathbb{E}[|X-\lambda|]  \leq 2 \frac{\lambda}{\sqrt{2\pi \lfloor\lambda\rfloor}} = 2 \frac{\lfloor\lambda\rfloor+\{\lambda\}}{\sqrt{2\pi \lfloor\lambda\rfloor}}  = \sqrt{\frac{2}{\pi}\lfloor\lambda\rfloor} + \{\lambda\}\times \sqrt{\frac{2}{\pi \lfloor\lambda\rfloor}}\,.
    \]

    For the lower bound
    \[
    \mathbb{E}[|X-\lambda|]  \geq 2 \lambda \frac{\lfloor\lambda\rfloor^{\lfloor\lambda\rfloor}}{(\lfloor\lambda\rfloor)!}e^{-\lfloor\lambda\rfloor} e^{-\{\lambda\}} \geq 2e^{-1}\lambda \frac{\lfloor\lambda\rfloor^{\lfloor\lambda\rfloor}}{(\lfloor\lambda\rfloor)!}e^{-\lfloor\lambda\rfloor} \,.
    \]
    Again, by the classical inequalities for Stirling formula~\cite{Rob55} 
    \[
    \frac{\lfloor\lambda\rfloor^{\lfloor\lambda\rfloor}}{(\lfloor\lambda\rfloor)!}e^{-\lfloor\lambda\rfloor} \geq \frac{1}{\sqrt{2\pi \lfloor\lambda\rfloor}} e^{-1/(12(\lfloor\lambda\rfloor+1))} \geq \frac{e^{-1}}{\sqrt{2\pi \lambda}} \,.
    \]
    Thus the inequality follows.
\end{proof}

\subsection{Expected window size: Proposition~\ref{prop:prop-window-2}}
\label{sec:proof-window}

We are now ready to finish the proof of Proposition~\ref{prop:prop-window-2}
\propwindowmapping*
\begin{proof}

We consider the cases separately $\lambda \leq \log k$ and $\lambda> \log k$ separately.

    -- Let us consider first the case $\lambda \leq \log k$.

    It is clear that we also have:
    \[
    W_{k,\lambda} \leq 2 + \max(X_1,\ldots,X_k)+\lambda\,.
    \]
    Thus, by \autoref{cor:max-k-poiss} and $\lambda\leq \log k$,
    \[
    \E[W_{k,\lambda}] \leq 2 + \E[\max(X_1,\ldots,X_k)] + \lambda \leq 2 + 3 \log k + \log k\,,
    \]
    and the inequality for this case follows.

    -- Let us now consider the case $\lambda > \log k$. We recall that $W_{k,\lambda}$ satisfies, 
    \[
    W_{k,\lambda} \leq 2 + \max(X_1,\ldots,X_k)-\min(X_1,\ldots,X_k)+|X_1-\lfloor\lambda\rfloor|\,.
    \]

    We distinguish according to whether $\lambda\geq 1$ or not. If ever $\lambda < 1$, we have $|X_1-\lfloor\lambda\rfloor|=X_1$ and
        \[
    \E[W_{k,\lambda}] \leq 2 + \E[\max(X_1,\ldots,X_k)-\min(X_1,\ldots,X_k)]+\E[X_1] \,.
    \]
    Now $\mathbb{E}[X_1]=\lambda < 1$, thus $\E[W_{k,\lambda}] \leq 3 + \E[\max(X_1,\ldots,X_k)-\min(X_1,\ldots,X_k)] \,,$ and we deduce
    \[
    \E[W_{k,\lambda}] \leq 3 + 4 \sqrt{\lambda \log k}\,.
    \]
    
    Finally,  if $\lambda \geq 1$ we apply Proposition~\ref{prop:diff-max-min-k-poisson} and Proposition~\ref{prop:deviation-mean-poisson},
    \[
    \E[W_{k,\lambda}] \leq 2 + 4 \sqrt{\lambda \log k} + \sqrt{\frac{2}{\pi}\lambda} + \sqrt{\frac{2}{\pi}}\,.
    \]
    Observe that, since $k\geq 2$, we have $\log k \geq \log 2 > \frac{2}{\pi}$ and $1>\frac{2}{\pi}$ so that
    \[
    \E[W_{k,\lambda}] \leq 3 + 5 \sqrt{\lambda \log k} \,.
    \]
    Combining both inequalities the proof is complete.
\end{proof}

\subsection{Expected overall cost of Algorithm~\ref{fig:procedure-final-mappings}}
\label{sec:proof-main-rm}

This section is dedicated to the proof of Theorem~\ref{thm:random-mapping-2}.

\thmrmmainrandommappingii*

The depth of recursion is clearly $O(\log k)$. What is not immediately clear is that the acceptance probability and the expected cost of each sampling round remain roughly the same throughout each recursive call.

Let us begin with the acceptance probability. When the algorithm is called with $(n,k)$, the probability that it accepts $X_1,\ldots,X_{k_\ell}$ and continues recursively is \[a'(n,k)=\frac{\Pr(X_1+\ldots+X_k = n)}{q} = \frac{\Pr(\Poiss{n}=n)}{\Pr(\Poiss{k_r \omega}=\lfloor k_r\omega\rfloor)}\,.\]

\begin{proposition}
    \label{prop:up-scaled-proba-mappings}
    The acceptance probability $a'(n,k)$ is bounded by below by $\frac{\sqrt{2}}{2}\times (1-\tfrac{2.5}{n})\geq 0.1$ for all $n\geq 3$ and $k\geq 1$.
\end{proposition}

\begin{proof}
    Remark that 
    \[
    p=\Pr(X_1+\ldots+X_k = n) = D(\omega\cdot k,n) = D(n,n) = \frac{n^n e^{-n}}{n!}\,,\quad q=D(\omega\cdot k_r,\lfloor \omega\cdot k_r\rfloor)\,.
    \]
    Observe  $D(a,\lfloor a \rfloor ) = D(\lfloor a \rfloor ,\lfloor a \rfloor ) (1+\{a\}/\lfloor a \rfloor)^{\lfloor a\rfloor} e^{-\{a\}}\leq  D(\lfloor a \rfloor ,\lfloor a \rfloor ) \leq \frac{1}{\sqrt{2\pi \lfloor a\rfloor}}$.

    Note that $k_r=k-\lfloor k/2\rfloor \geq k/2$, whence $\omega\cdot k_r \geq n/2$. Thus $q\leq \frac{1}{\sqrt{2\pi \lfloor \omega\cdot k_r\rfloor}} \leq  \frac{1}{\sqrt{2\pi \lfloor \tfrac{n}{2}\rfloor }}$, while $p=D(n,n)\geq \tfrac{1}{\sqrt{2\pi n}} \exp(-\tfrac{1}{12 n})$ holds due to Stirling (see e.g., \cite{Rob55}).
    
    We deduce the bound
    \[
    \frac{p}{q}\geq  \sqrt{\frac{\lfloor n/2\rfloor}{n}} \exp(-\tfrac{1}{12 n}) \geq \frac{\sqrt{2}}{2}\sqrt{1 - \frac{2}{n}} \exp(-\tfrac{1}{12 n}) \,,
    \]
    holds for all $n\geq 2$. As $\sqrt{1 - \frac{2}{n}}\geq 1-\tfrac{2}{n}$, $\exp(-\tfrac{1}{12 n})\geq 1-\tfrac{1}{12 n}$ and $(1-\tfrac{2}{n})(1-\tfrac{1}{12n})\geq 1 - \tfrac{2+1/12}{n}$, we are done because $2+1/12\leq 2.5$.
\end{proof}

\begin{remark}
    Moreover, when $k$ is even, or $k=k(n)\to\infty$, we have that $a'(n,k)\to \sqrt{2}/2$.
\end{remark}

This means that the number of rejections, regardless of $n$ and $k$, is expected to be constant. It remains to study the cost of the sampling rounds. For \emph{deterministic} $n$ and $k$, a sampling round costs $\mathbb{E}[S_{n,k}] \leq 5\sqrt{\tfrac{n}{k}\log k} + 4 \log k + 3$ due to Prop.~\ref{prop:prop-window-2}. In our case, as the recursion proceeds we obtain $N_0=n$, $N_1$, $N_2$, and so on, that are random variables (except for $N_0$ that is constant). On the other hand $k_{i+1}=k_i-\lfloor k_i/2\rfloor$, $k_0=k$, are deterministic.

Still, we can prove that $\mathbb{E}[\nicefrac{N_i}{k_i}]=\nicefrac{n}{k}$ for each $i$. This, together with the concavity of $x\mapsto \sqrt{x}$ imply the following proposition:
\begin{restatable}{proposition}{expectedtimerecurrence}
\label{prop:recurrence-nlogn}
    For each recursive call $(N_i,k_i)$ in Algorithm~\ref{fig:procedure-final-mappings}, the expected cost of the sampling rounds satisfy
$\mathbb{E}[S_{N_i,k_i}] \leq 5\sqrt{\tfrac{n}{k}\log k} + 4 \log k+ 3$.
\end{restatable}

\begin{remark}
   On the regime $k\leq \tfrac{n}{\log n}$,  Proposition~\ref{prop:prop-window-2} implies that $\mathbb{E}[S_{n,k}] \leq 5\sqrt{\tfrac{n}{k}\log k} + 3$. Still, it is possible for recursive calls $(N_1,k_1),(N_2,k_2)\ldots$ to fall outside of this regime. That is why we use the most general $\mathbb{E}[S_{N_i,k}] \leq 5\sqrt{\tfrac{N_i}{k_i}\log k_i} + 4 \log k_i + 3$, implied by Proposition~\ref{prop:prop-window-2}.
\end{remark}

In our analysis we count the number of times we sample from a simple random variable such as Poisson, Binomial or Bernoulli.

Still, we remark that summing the profiles is not costly. 
Referring to the Algorithm~\ref{fig:procedure-final-mappings}, the expected length $\mathbb{E}[{\tt Len}(profile_\ell)]$  is bounded by the expected value of the sum of the costs of all costs of the sampling rounds effectuated before accepting $profile_\ell$. Since the number of sampling rounds is expected to be constant, by Wald's equation we deduce $\mathbb{E}[{\tt Len}(profile_\ell)] = O(\sqrt{\tfrac{n}{k}\log k} + \log k)$.  If $k$ is even, by symmetry $A=(X_1,\ldots,X_{k/2})$ and $B=(X_{k/2+1},\ldots,X_k)$, are identically distributed, even given that the sum is $n$. Thus $\mathbb{E}[{\tt Len}(profile_r)]=O(\sqrt{\tfrac{n}{k}\log k} + \log k)$. If $k$ were odd, the argument is easily adapted. This implies that the sums cost, overall  $O((\log k)(\sqrt{\tfrac{n}{k}\log k} + \log k))$. As a by-product, we have proved that the expected length of the profile of a random mapping is  $O(\sqrt{\tfrac{n}{k}\log k} + \log k)$. To achieve this complexity, we represent the profiles by two dynamic arrays that contains in sequence the value to the right (resp. to the left) of the mode $\lfloor \frac{n}{k} \rfloor$.

\section{Length of the profile of a random mapping}
\label{sec:length-profile}

This section is devoted to the proof of Proposition~\ref{prop:length-profiles-random}. 
\proplengthprofilesrandommapping*

\begin{proof}
Without loss of generality, up to adjusting the constants,
we may assume that $n\geq3$ and $k>2$.
Let $\lambda:=\omega=n/k$ and $m:=\lfloor k/2\rfloor$.
Let $C_{m,\lambda}$ denote the length of the profile of
$m$ iid $\Poiss{\lambda}$ random variables.

Algorithm~\ref{fig:procedure-final-mappings} can be viewed
as a rejection sampler for the half-profile of a random
mapping, using profiles of $m$ iid $\Poiss{\lambda}$
random variables as proposals. For this, we omit the initial
test for small $k$ and stop as soon as the first half-profile
is accepted. Its acceptance probability
is uniformly bounded below by a constant $\delta>0$,
by Proposition~\ref{prop:up-scaled-proba-mappings}.

We first consider the case $k\geq n^{1/3}$ (or equivalently
$k\geq\sqrt{\lambda}$).

For the lower bound, since $k\geq\sqrt{\lambda}$ and
$m=\lfloor k/2\rfloor\geq k/3$,
Theorem~\ref{thm:optimal-length-profile}, applied with
$d=1/3$, shows that $C_{m,\lambda}$ is $\Omega(\sqrt{\lambda})$ whp as
$\lambda\to\infty$, uniformly over this regime.
By Lemma~\ref{lemma:transfer-lower-bound-ex}, this gives
$\E[\length{n}{k}]=\Omega(\sqrt{\lambda})$, uniformly
over this regime. Together with the trivial bound
$\length{n}{k}\geq1$, this yields
$\E[\length{n}{k}]
    \geq c_2(1+\sqrt{\lambda})
    =c_2\left(1+\sqrt{\frac nk}\right)$
for some constant $c_2>0$.

For the upper bound,
Proposition~\ref{prop:expected-length-iid-poisson} gives
\[
    \E[C_{m,\lambda}]
    \leq1+4\sqrt{\lambda\log m}+3\log m
    \leq1+4\sqrt{\lambda\log k}+3\log k.
\]
By Lemma~\ref{lemma-upper-bound-wald}, the expected
length of the accepted half-profile is at most
$\delta^{-1}\E[C_{m,\lambda}]$.
If $k$ is even, the other half has the same distribution
by symmetry. If $k$ is odd, the expected length of the
profile of the other half differs from that of the first
half by at most one.
Since the length of the combined profile is at most
the sum of the lengths of the two half-profiles,
\[
    \E[\length{n}{k}]
    \leq\frac{2}{\delta}\E[C_{m,\lambda}]+1
    \leq\frac{2}{\delta}
        \bigl(1+4\sqrt{\lambda\log k}+3\log k\bigr)+1.
\]
There exists a constant $c_3>0$ such that
$\E[\length{n}{k}]\leq c_3
\left(\sqrt{\frac nk\log k}+\log k\right)$.

We now consider $k\leq n^{1/3}$, equivalently
$k\leq\sqrt{\lambda}$.
The upper bound is immediate.
For the lower bound, part~(i) of Proposition~\ref{TODO} gives
$C_{m,\lambda}=\Omega(k)$ whp as $k\to\infty$,
uniformly in $\lambda$.
As above, using Lemma~\ref{lemma:transfer-lower-bound-ex},
we conclude that there exists a constant $c_1>0$ such that
$c_1k\leq\E[\length{n}{k}]\leq k$.
\end{proof}

\section{The acceptance probability of the exact-size Boltzmann sampler for surjections}

\label{appendix:properties-omega-surjections}

In this section we prove the following theorem regarding the acceptance probability of the exact-size Boltzmann sampler for surjections.

\acceptanceprobasurjection*

First we need to characterize the optimal parameter $\omega=\omega(n,k)$ 
used in the Boltzmann sampler. The following proposition characterizes the main properties we will need, for the proof of Theorem~\ref{thm:order-acceptance-prob-surj}, as well as later on.

\propertiesomega*
\begin{proof}
    Write $\Delta=n-k$ and note  \[
        k\omega - \Delta = k\times \frac{e^\omega - 1 - \omega}{e^\omega - 1}\,.
    \]
    Clearly the right-hand side is positive, which implies $k\omega \geq \Delta$. Observe that this is equivalent to $\omega \geq \tfrac{n}{k}-1>0$.
    
    On the other hand $e^\omega - 1 - \omega \leq \tfrac{1}{2}\omega \times (e^\omega - 1)$, just by comparing coefficients. Indeed,
    \[e^\omega-1-\omega = \tfrac{\omega^2}{2} + \tfrac{\omega^3}{3!}+\tfrac{\omega^4}{4!}+\ldots\leq \tfrac{\omega}{2} \omega + \tfrac{\omega}{2} \tfrac{\omega^2}{2!} + \tfrac{\omega}{2} \tfrac{\omega^3}{3!}+\ldots  = \tfrac{\omega}{2}\times (e^\omega - 1)\] follows by comparing term by term.
    Thus
    \[
    k\omega - \Delta \leq \tfrac 1 2 k\omega\,,
    \]
    which shows that $\Delta \geq \tfrac{1}{2}k \omega$.

    For the last inequality we remark that $\frac{e^\omega - 1 - \omega}{e^\omega - 1} \geq \tfrac{\omega}{2} - \tfrac{\omega^2}{12}$ by a similar argument.   The result follows upon multiplying by $2k$ since $k\omega - \Delta = k\times \frac{e^\omega - 1 - \omega}{e^\omega - 1}$.
\end{proof}

For the proof of \autoref{thm:order-acceptance-prob-surj}, we notice that the estimate for $a(n,k)$ equivalent to an estimate of $\stirlingsecond{n}{k}$. Indeed,  $\tfrac{k!}{n!}\stirlingsecond{n}{k} =   [z^n] (e^z-1)^k$, therefore $\frac{\omega^n}{(e^\omega-1)^k} \tfrac{k!}{n!}  \stirlingsecond{n}{k}= [z^n] \left(\frac{e^{\omega z}-1}{e^\omega -1}\right)^k = a(n,k)$.
We now use an estimate for $\stirlingsecond{n}{k}$ form Chelluri, Richmond and Temme~\cite{ChelluriRichmondTemme}. We cite here the result from~\cite{ChelluriRichmondTemme}.  This involves $\omega=\omega(n,k)$, which is called $u_0$ in~\cite{ChelluriRichmondTemme}. 

\begin{theorem}
    \label{thm:estimate-stirling-2}
    Uniformly on  $1\leq k \leq n - n^{1/3}$ we have 
    \begin{equation}
        \stirlingsecond{n}{k} = \frac{n!}{k!} \frac{(e^\omega-1)^k}{2 \omega^n} \frac{1}{\sqrt{\pi \omega k H_0(\omega)}} \left(1+O(n^{-1})\right)\,,\qquad 2H_0(u) = \frac{e^u}{e^u-1}\left(1-\frac{u}{e^u-1}\right)\,.
    \end{equation}
    If $n - n^{1/3}\leq k \leq n$, uniformly we have
    \begin{equation}
        \stirlingsecond{n}{k} = \frac{1}{2^{n-k}} \frac{n^{2(n-k)}}{(n-k)!} \left(1+O(n^{-1/3})\right)\,.
    \end{equation}
\end{theorem}

\begin{proof}[Proof of~\autoref{thm:order-acceptance-prob-surj}]
Consider the case $k\leq n-n^{1/3}$. We write
\[
a(n,k) = \frac{\omega^n}{(e^\omega-1)^k} \frac{k!}{n!}  \stirlingsecond{n}{k}= \frac{1}{\sqrt{4 \pi \omega k H_0(\omega)}} \left(1+O(n^{-1}\right)))\,.
\]
Observe that $\frac{e^\omega}{e^\omega-1} = \frac{n}{k \omega}$  and so $2H_0(\omega) = \frac{n}{k \omega} \left(1-\tfrac{\omega}{e^\omega - 1}\right)$. Simplifying we obtain the stated estimate, but with an error $O(n^{-1})$ which is clearly $O(\Delta^{-1})$. 

Consider now the case $n-n^{1/3}\leq k\leq n$. We show a slightly more precise expression \begin{equation}
    a(n,k) = \frac{\Delta^\Delta e^{-\Delta}}{\Delta!} (1+O(n^{-1/3}))\,,
\end{equation} uniformly. This implies our result as $\frac{\Delta^\Delta e^{-\Delta}}{\Delta!} = \frac{1}{\sqrt{2\pi \Delta}} (1+O(1/\Delta))$ and $n(1-\tfrac{\omega}{e^\omega-1})=\Delta+O(n^{-1/3})$. We explain the latter equation below.

Proposition~\ref{prop:properties-omega-surjections} tells us that $0\leq 2\Delta-k\omega \leq \tfrac{1}{6}k\omega^2 $ and $\omega \leq \frac{2 \Delta}{k}$, thus $0\leq 2\Delta-k\omega\leq \tfrac{1}{6} k (2\Delta/k)^2 = \tfrac{2}{3} \frac{\Delta^2}{k}$, which is $O(n^{2/3}/k) = O(n^{-1/3})$. We deduce that $\omega=\frac{2\Delta}{k}+O(n^{-4/3})$.

Write
\[
a(n,k) = \frac{\omega^n}{(e^\omega-1)^k} \frac{k!}{n!}  \frac{1}{2^{\Delta}} \frac{n^{2\Delta}}{\Delta!}\left(1+O(n^{-1/3})\right) = \left(\tfrac{\omega}{e^\omega-1}\right)^k \left(\tfrac{k!}{n!}n^\Delta\right) \frac{(\omega n / 2)^\Delta}{\Delta!}\left(1+O(n^{-1/3})\right)\,.
\]
We will now study each of the factors. First $\frac{\omega}{e^\omega-1}=1-\tfrac{\omega}{2}+O(\omega^2)$, uniformly, so that \[\left(\tfrac{\omega}{e^\omega-1}\right)^k = \exp(-k\omega/2+O(k\omega^2)) = \exp(-\Delta+O(n^{-1/3})) = e^{-\Delta}(1+O(n^{-1/3}))\,.\] 

Next 
\begin{align*}
    \left(\tfrac{k!}{n!}n^\Delta\right)  &= \frac{n^\Delta}{n (n-1)\ldots (n-\Delta+1)} \\
    &= \frac{1}{\exp(O(1/n+\ldots+(\Delta-1)/n))} \\
    &= \exp(O(\Delta^2/n)) = 1 + O(n^{-1/3})\,.
\end{align*}

Observe that $\omega n / 2 = \omega k / 2+ \omega \Delta/2 = \Delta + O(n^{-1/3})$ and therefore $(\omega n / 2)^\Delta = \Delta^\Delta (1+O(n^{-1/3})$. Finally, we note that $\Delta\leq n^{1/3}$ and therefore we may substitute $O(n^{-1/3})$ by $O(\Delta^{-1})$.

Finally we explain why $n(1-\tfrac{\omega}{e^\omega-1})=\Delta+O(n^{-1/3})$. Note that $1-\tfrac{\omega}{e^\omega  - 1} = \tfrac{\omega}{2}+O(\omega^2)$ as before, while $n=k+\Delta$, thus $n(1-\tfrac{\omega}{e^\omega-1}) = \frac{k \omega}{2}+O(\Delta \omega) = \Delta + O(n^{-1/3})+O(\Delta\omega)$. Since $\Delta\omega = O(\Delta^2/k) = O(n^{-1/3})$ in this regime, the assertion follows.
\end{proof}

\begin{remark}
    We can prove that $a(n,k) = \frac{\Delta^\Delta}{\Delta!}e^{-\Delta} (1+O(n^{-1/3}))$ when $\Delta \leq n^{1/3}$. In particular, the probability tends to a constant if $\Delta$ is constant.
\end{remark}

\section{Proofs regarding the up-scaling of the acceptance probability in Algorithm~\ref{alg:surjection-upper-upper} and Algorithm~\ref{alg:surjection-upper-complete} for surjections}
\label{appendix-proof-sampling-rounds-upper}

In this section we complete the analysis of Algorithm~\ref{alg:surjection-upper-upper}, which is also need in order to complete the analysis of Algorithm~\ref{alg:surjection-upper-complete}, the full algorithm for surjections in the region $k\geq \tfrac{n}{\log n}$. The complete pseudo-code can be found below in Section~\ref{section:complete-pseudocode-surjections}.

\subsection{General considerations}

We recall that the scaling factor is $q$ is related to the maximum of the Poisson distribution. In order to prove the results we will need a uniform estimate. 

\begin{lemma}
    \label{lemma:max-poisson-estimate}
    Let $\lambda > 0$. The probability $p(\lambda)=\Pr(\Poiss{\lambda}=\lfloor\lambda\rfloor)$ satisfies $p(\lambda)=\Theta((1+\lambda)^{-1/2})$ uniformly on $\lambda>0$.

    In particular, fix $\delta>0$, uniformly on $\lambda \geq \delta$ we have $p(\lambda)=\Theta(\lambda^{-1/2})$ where the constants depend on $\delta$.
\end{lemma}
\begin{proof}
    When $\lambda\geq 1$, this clearly follows from Lemma~\ref{lemma:mode-poisson} as $\lfloor\lambda\rfloor = \Theta(\lambda+1)$ when $\lambda \geq 1$, indeed $\lfloor \lambda\rfloor\leq \lambda+1\leq \lfloor\lambda\rfloor+2\leq 3 \lfloor\lambda\rfloor$. For $\lambda < 1$, $p(\lambda) = \Pr(\Poiss{\lambda}=0)= e^{-\lambda} \in (e^{-1},1)$, while clearly $\tfrac{\sqrt{2}}{2}\leq (1+\lambda)^{-1/2}\leq 1$. Thus we can find constants such that $p(\lambda)=\Theta((1+\lambda)^{-1/2})$ in this case too. Since we have only two cases, $\lambda \geq 1$ and $\lambda<1$, the result follows.
\end{proof}

\begin{remark}
    Essentially, writing $\Theta((1+\lambda)^{-1/2})$ is a more convenient way of saying that a probability can be at most $p(\lambda)\leq 1$.
\end{remark}

\subsection{Proof of Proposition~\ref{prop:sampling-round-upper-upper}}

In the case of Proposition~\ref{prop:sampling-round-upper-upper}, the scaling factor $q$ is given by
\[
q=\Pr(\Poiss{kp_1}=\lfloor kp_1\rfloor)\times \Pr(\Poiss{kp_2}=\lfloor kp_2\rfloor)\,.
\]

Following Lemma~\ref{lemma:max-poisson-estimate}, there are essentially two cases when estimating $\Pr(\Poiss{kp_j}=\lfloor kp_j\rfloor)$: either $kp_j \geq \delta$ for some fixed $\delta > 0$, in which case we can directly write $\Pr(\Poiss{kp_j}=\lfloor kp_j\rfloor)=\Theta((kp_j)^{-1/2})$, or $k p_j$ is small and the probability is rather $\Theta(1)$. The latter case means that $q$ is a \emph{bad} scaling factor; it does not improve the acceptance probability or the expected number of sampling rounds.

The following lemma proves that $k p_1$ and $k p_2$ are at least constant when $k\geq \frac{n}{\log n}$. We notice that $k p_1$ is the expected number of elements in $[k]$ with one preimage, and similarly $k p_2$ is the expected number of elements with two. 

\begin{lemma}
    Consider $n>k\geq 1$ such that $k\geq \frac{n}{\log n}$ and let $\omega=\omega(n,k)$ satisfy the saddle-point equation $k\omega \frac{e^\omega}{e^\omega-1} = n$.
    Let $p_1(\omega)=\Pr(\Poissone{\omega}=1)=\frac{\omega}{e^\omega-1}$ and $p_2(\omega)=\Pr(\Poissone{\omega}=2) =\frac{\omega^2/2}{e^\omega-1}$. 
    
    Then, uniformly on $n,k$ with the above conditions, $kp_1(\omega)\geq 1$ and $kp_2(\omega)\geq \tfrac{1}{4}$ .
\end{lemma}
\begin{proof}
    The case of $p_1=p_1(\omega)$ is simple. We have $k p_1 =n e^{-\omega}$ due to the saddle-point equation. But $\omega \leq \frac{n}{k}\leq \log n$, therefore $n e^{-\omega}\geq n e^{-\log n}=1$.

    For $k p_2$, the statement is equivalent to a stronger version of the previous inequality: $n e^{-\omega}\geq \frac{1/2}{\omega}$. Indeed, this inequality holds iff $\log n \geq \omega -\log \omega - \log 2$. Remark that $t\to t-\log t$ decreases for $t\in (0,1)$, increases for $t\geq 1$ and is $0$ for $t=1$. When $\omega \geq 1$, $\omega\leq n/k \leq  \log n$ implies\footnote{Remark that $\log \log n +\log 2$ is positive even for $n=2$, and $n\geq 2$ holds because $n>k\geq 1$.} $\omega -\log \omega- \log 2 \leq \log n -\log \log n -\log 2 < \log n$. While for $\omega<1$ we have $\omega\geq \frac{\Delta}{k}\geq \frac{1}{k}\geq \frac{1}{n}$, due to Prop.~\ref{prop:properties-omega-surjections}, implying $\omega-\log \omega - \log 2 \leq \tfrac{1}{n}+\log n - \log 2$. Since $n>k\geq 1$, we have $n\geq 2$ and so $\tfrac{1}{n}-\log 2\leq \tfrac{1}{2}-\log 2<0$. We deduce $\omega-\log \omega-\log 2 < \log n$.
\end{proof}

\clearpage
\subsection{Complete algorithm for surjections in the region $k\geq \tfrac{n}{\log n}$}
\label{section:complete-pseudocode-surjections}

Here we provide the full pseudo-code for the reader.

  \begin{center}
      \begin{algorithm}[H] \scriptsize               
      \caption{\label{alg:surjection-upper-complete}Algorithm for a random surjection profile in the regime $k\geq \frac{n}{\log n}$. 
      } 
\label{fig:procedure-surjections-large-k}
         function RandomSurjetionProfile($n$,$k$):\\

    $\omega$ := SaddlePoint($n,k$)\;
    $a $ := $\max(\lfloor \omega\rfloor,1)$\;
    $b $ := $a + 1$\;
    $p_a$ := $\Pr(\Poissone{\omega}=a)$ \;
    $p_b$ := $p_a \times \frac{\omega}{b}$ \;
    $q$ := $ \Pr(\Poiss{k p_a}=\lfloor k p_a\rfloor) \times 
     \Pr(\Poiss{k p_b}=\lfloor k p_b\rfloor)$\;
    $p_r$ := $\Pr(\Poiss{\omega}\geq b + 1) / \Pr(\Poiss{\omega}\geq 1)$\;
    \While{{\tt True}}{
        $profile$ := $\{\}$\;
        $R \gets {\tt Poiss}(k\cdot p_r)$\;
        $L := k - R$\;
        $N$ := $n$ \;
        \For{$j= 1$ \KwTo $a-1$}{
            $p_j$ := $\Pr(\Poissone{\omega}=j)$\;
             $m \leftarrow {\tt Poiss}(k\cdot p_j)$\;
             add $(j,m)$ to profile {\bf if} $m>0$\;
            $L := L - m$\;
            $N := N - m\cdot j$\;
            $j := j+1$\;
    }
        $j$ := $3$\;
        \While{$R>0$}{
             $r_j := {\footnotesize\Pr(\Poiss{\omega}=j  |  \, \Poiss{\omega}\geq j)}$\;
            $m \leftarrow {\tt Bin}(R,r_j)$\;
             add $(j,m)$ to profile {\bf if} $m>0$\;
            $R := R - m$\;
            $N := N - m\cdot j$\;
            $j := j+1$\;
        }

        $m_a$ := $bL - N$\;
        $m_b$ := $N - a L$ \;
        \If{$m_a\geq 0$ and $m_b\geq 0$}{
            $p $ := $\Pr(\Poiss{k p_a}=m_a)\times \Pr(\Poiss{kp_b}=m_b)$\;
            \If{Bernoulli(p/q)}{
                
                add $(a,m_a)$ to profile {\bf if} $m_a>0$\;
                add $(b,m_b)$ to profile {\bf if} $m_b>0$\;
                \Return $profile$\;
            }
         }
    }
    
      \end{algorithm}
      \end{center}

  \subsection{Proof of Theorem~\ref{thm:sampling-rounds-algo-upper-2}}
  \label{sec:proof-theorem-upper-complete}

  \thmsamplingroundsalgouppercomplete*
  \begin{proof}
      We may assume that $\omega \geq 2$. Indeed, if $\omega<2$ then $a=1$ is the mode and we obtain Algorithm~\ref{alg:surjection-upper-upper}, which requires $\Theta(1)$ when $\omega<2$. By Prop.~\ref{prop:properties-omega-surjections} we obtain $1<n/k<3$ and the assertion holds.

Consider first
$p_a:=\Pr(\Poissone{\omega}=\lfloor\omega\rfloor) = \Pr(\Poiss{\omega}=\lfloor\omega\rfloor) \frac{e^\omega}{e^\omega-1} .$ As we are supposing $\omega\geq 2$, we may apply Lemma~\ref{lemma:max-poisson-estimate} and $p_a=\Theta(\omega^{-1/2})$ uniformly when $\omega\geq 2$.  Since also $\omega=\Theta(n/k)$ uniformly on $\omega\geq 2$ by Prop.~\ref{prop:properties-omega-surjections},  $k p_a =\Theta(k \sqrt{\nicefrac{k}{n}})$. Let $b=a+1$. Similarly $p_b = p_a \frac{\omega}{\lfloor \omega\rfloor+1}$ and $k p_b =\Theta(k \sqrt{\nicefrac{k}{n}})$.  We conclude that these are $\Omega(1)$ when the conditions $\omega=\omega(n,k)\geq 2$ and $k \sqrt{\nicefrac{k}{n}} \geq 1$ are verified. The latter is equivalent to $k\geq n^{2/3}$ which is verified by hypothesis.

 Thus we may apply Lemma~\ref{lemma:max-poisson-estimate}, and the rescaling factor is $q = \Theta(\nicefrac{1}{\sqrt{kp_a kp_b}}) = \Theta(\sqrt{\nicefrac{n}{k^3}}).$ When $\omega\geq 2$, $\omega=\Theta(n/k)$ due to Prop.~\ref{prop:properties-omega-surjections}. This means that $\hat{a}=\Theta(\nicefrac{1}{\sqrt{k^2\omega}})=\Theta(\nicefrac{1}{\sqrt{k n}})$. Thus the rescaled probability is $\frac{\hat{a}}{q} = \Theta(\nicefrac{\sqrt{k^3/n}}{\sqrt{k n}}  ) = \Theta(\tfrac{k}{n})$ and the result follows.
  \end{proof}

\section{Length of the profile of a random surjection}
\label{sec:length-profile-surjection}

In this section we show how to deduce the order of the profile of a 
random surjection from that of random mappings. We prove the following
\begin{proposition}
    \label{prop:length-surjections-from-mappings}
    Let $L_{n,k}$ be the length of the profile of a random mapping from $[n]$ to $[k]$, and let $R_{n,k}$ be the length of the profile of a random surjection from $[n]$ to $[k]$. On the region $k\leq \tfrac{n}{\log n}$, we have $\mathbb{E}[R_{n,k}] = \Theta(\mathbb{E}[L_{n,k}])$ uniformly. 
\end{proposition}
We remark that we do not consider the region $k\geq \tfrac{n}{\log n}$ as the algorithm we consider is already $\tilde{O}(1)$ in complexity, and then so is the length of the profile of a random surjection.
\begin{proof}
Recall that the probability $p(n,k)$ that a random mapping from $[n]$ to $[k]$ is a surjection satisfies $p(n,k)\geq 1 - \tfrac{1}{\log n}$ by Lemma~\ref{lemma:coupon-collector}. This will actually imply the result for surjections in this region.

-- The upper bound $\E[R_{n,k}] = O(\E[L_{n,k}])$ follows at once from Lemma~\ref{lemma-upper-bound-wald} and the fact that $p(n,k) \geq 1-\tfrac{1}{\log 3} >0$ already for $n\geq 3$.

-- The lower bound $\E[R_{n,k}] = \Omega(\E[L_{n,k}])$ follows from Lemma~\ref{lemma:transfer-lower-bound-ex} and Remark~\ref{remark:transfer-two-parameters}, since $p(n,k)\to 1$ uniformly in $k$.
\end{proof}

\end{document}